\documentclass{article}
\usepackage{ijcai23}

\usepackage{times}
\usepackage{soul}
\usepackage{url}
\usepackage[hidelinks]{hyperref}
\usepackage[utf8]{inputenc}
\usepackage[small]{caption}
\usepackage{graphicx}
\usepackage{amsmath}
\usepackage{amsthm}
\usepackage{amsfonts}
\usepackage{amssymb}
\usepackage{booktabs}
\usepackage{algorithm}
\usepackage{algorithmic}
\usepackage[switch]{lineno}
\usepackage{color, xcolor} 
\usepackage{subfigure}

\newtheorem{theorem}{Theorem}
\newtheorem{lemma}{Lemma}
\newtheorem{assumption}{Assumption}

\title{Towards Generalizable Reinforcement Learning for Trade Execution}

\author{
Chuheng Zhang\thanks{The authors have contributed equally to this work.}$^{1,2}$\and
Yitong Duan\footnotemark[1]$^2$\and
Xiaoyu Chen$^2$\and
Jianyu Chen$^2$\and
Jian Li$^2$ \And
Li Zhao$^1$
\affiliations
$^1$Microsoft Research\\
$^2$IIIS, Tsinghua University
\emails
zhangchuheng123@live.com,
\{dyt19, chen-xy21, lijian83\}@mails.tsinghua.edu.cn,
jianyuchen@tsinghua.edu.cn,
lizo@microsoft.com
}

\begin{document}

\maketitle

\begin{abstract}
    
Optimized trade execution is to sell (or buy) a given amount of assets in a given time with the lowest possible trading cost. Recently, reinforcement learning (RL) has been applied to optimized trade execution to learn smarter policies from market data.
However, we find that many existing RL methods exhibit considerable overfitting which prevents them from real deployment.
In this paper, we provide an extensive study on the overfitting problem in optimized trade execution.
First, we model the optimized trade execution as offline RL with dynamic context (ORDC), where the context represents market variables that cannot be influenced by the trading policy and are collected in an offline manner. 
Under this framework, we derive the generalization bound and find that the overfitting issue is caused by large context space and limited context samples in the offline setting.
Accordingly, we propose to learn compact representations for context to address the overfitting problem
, either by leveraging prior knowledge or in an end-to-end manner. 
To evaluate our algorithms, we also implement a carefully designed simulator based on historical limit order book (LOB) data to provide a high-fidelity benchmark for different algorithms. 
Our experiments on the high-fidelity simulator demonstrate that our algorithms can effectively alleviate overfitting and achieve better performance.
\end{abstract}

\section{Introduction}
\label{sec:trade_introduction}

Nowadays, brokerage firms are required to execute orders on behalf of their clients (e.g., retail or institutional investors) to ensure execution quality.
Optimized trade execution, whose objective is to minimize the execution cost of trading a certain amount of shares within a specified period, is an important task towards better execution quality.
In modern financial markets, most of the transactions are conducted through electronic trading.
Therefore, developing a smart agent for optimized trade execution in electronic markets is a critical problem in the financial industry.

Traditional solutions for this problem \cite{almgren2001optimal,gueant2012optimal,bulthuis2017optimal} usually make strong assumptions on the price or transaction dynamics and therefore do not apply to real scenarios. Moreover, these strategies are static (i.e., determined before the start of the trading) and therefore unable to adapt to the real-time market.
Recently, RL-based methods have been developed to learn a more adaptive agent from market data \cite{fang2021universal,ning2018double,lin2020deep}.
However, we find that existing methods suffer from considerable overfitting. 
As shown later in Figure~\ref{fig:merge} (right), trained models are prone to memorize the history instead of learning generalizable policies.

To better analyze overfitting in trade execution, we propose a framework called Offline Reinforcement learning with Dynamic Context (ORDC) to model the problem.
This framework highlights the difficulty in generalization for the trade execution task.
In trade execution, part of the observation (which we call \emph{context}) evolves independently of the agent's action, and the simulator is based on a dataset that contains a finite number of context sequences (e.g., historical price sequences). The number of context sequences is usually limited and does not increase w.r.t. the number simulation steps at training time. 
Therefore, the agent is prone to memorize the training context sequences, and perform not well on testing context sequences. We highlight the difficulty in generalization under this setting theoretically and show that limited context sequences lead to bad generalization.
This explains why generalization is hard for ORDC. The offline nature does not receive much attention in previous RL applications for trade execution that usually employ off-the-shelf online RL methods with data-driven simulation.

Since it is usually hard to increase the number of sampled context sequences for training in practice,
we find another way to address the overfitting problem. 
The theoretical analysis also indicates that larger context space leads to worse generalization under the same number of samples.
Motivated by the analysis, we propose to aggregate the context space by learning a compact context representation for better generalization.
This is effective for trade execution where the context usually contains more information than needed, but only a small amount of the underlying information is helpful for decision-making.
Moreover, we design a simplified trade execution task that motivates us to learn a compact context representation that is predictive of the statistics on future contexts.
Therefore, we propose to use the prediction of future statistics as the context representation.
We propose two algorithms:
\textbf{CASH} (\textbf{C}ontext \textbf{A}ggregate with \textbf{H}and-crafted \textbf{S}tatistics) which can learn interpretable models, and 
\textbf{CATE} (\textbf{C}ontext \textbf{A}ggregate with \textbf{E}nd-to-end \textbf{T}raining) which does not require domain knowledge.

In the experiment, we first implement a high-fidelity and open-source simulator for trade execution to provide a uniform backtest for different methods.
With this simulator, we find that previous state-of-the-art algorithms suffer from overfitting and our algorithm outperforms these baselines due to better generalization.

The contributions of this paper are as follows:
\begin{itemize}
    \item (Section \ref{sec:why_difficult}) We propose the ORDC framework and provide theoretical analysis to highlight the difficulty in generalization for the trade execution task. 
    \item (Section \ref{sec:better_generalization}) We propose two algorithms to learn generalizable representations for trade execution. One is interpretable with the help of human prior and the other learns {in an end-to-end manner.}
    \item (Section \ref{sec:execution_experiment}) We implement an open-source, high-fidelity, and flexible simulator to reliably compare different trade execution algorithms. With this simulator, we show that our models learn more generalizable policies and outperform previous methods.
\end{itemize}

\section{Related Work}

Most existing papers on RL generalization study under the contextual MDP setting where the agent is trained and evaluated on different sets of configurations \cite{zhang2018study,zhang2018dissection,packer2018assessing} or procedurally generated environments \cite{cobbe2019quantifying,cobbe2020leveraging,song2019observational,wang2020improving}.
ORDC is different from their settings in that 1)~the context changes in each time step and affects both the reward and the transition; 2)~the context sequence is highly stochastic and pre-collected with limited volume.
These properties contribute to the difficulty in estimating the value function or evaluating the policy and thus exacerbate overfitting.
A recent survey \cite{kirk2021survey} points out that benchmarking RL algorithms with popular procedurally generated environments is not enough and RL generalization in other settings (e.g., the offline setting) is valuable and under-explored.

Not only limited to trade execution, the structure of ORDC is common for many industrial RL application scenarios such as video stream control \cite{mao2017neural}, inventory management \cite{oroojlooyjadid2022deep}, ride-sharing \cite{shen2020auxiliary}, cellular network control \cite{dietterich2018discovering}, etc.
Although ORDC emphasizes the offline nature, it is also different from the canonical offline RL setting.
The agent is trained to avoid encountering unseen states in offline RL, whereas the agent in ORDC is evaluated on unseen context sequences which cannot be avoided.
Therefore, off-the-shelf offline RL algorithms do not apply to our setting and a more adaptive algorithm is needed.
Many previous solutions for trade execution use online RL algorithms to learn a policy from the interaction with the data-driven simulator.
However, unlike online RL settings where the testing and training environments are the same, ORDC tests the agent on unseen context sequences which brings in difficulty in generalization.
IDSD \cite{shahamiriname2008reinforcement} and input-driven MDP \cite{mao2018variance} are similar to ORDC in that they model the uncontrollable part in observation. However, they focus on the online setting and do not model the existence of a smaller latent context space (see section \ref{sec:why_difficult}).

\section{Why Generalization is Difficult for Trade Execution?}
\label{sec:why_difficult}

In this section, we first briefly introduce the trade execution problem. Then, we introduce Offline RL with Dynamic Context (ORDC) which models how RL is used to solve the problem. At last, we provide theoretical analysis for the ORDC model to highlight the difficulty in generalization for these RL applications.

\subsection{Trade Execution}
Modern electronic markets match the buyers and sellers with the limit order book (LOB), which is a collection of outstanding orders, each of which specifies the direction (i.e., buy or sell), the price, and the volume. The traders can trade via two types of orders: market orders (MOs) and limit orders (LOs).
An MO is executed immediately but may suffer from a large trading cost (e.g., due to crossing the spread or temporary market impact \cite{almgren1999value}), whereas an LO can provide the trader with a better price but at the risk of non-execution. See appendix for more details on LOB. 
Moreover, the price fluctuation or trend can also affect the trading cost.
Optimized trade execution aims to buy/sell a given amount of assets in a given time period at a trading cost as low as possible. For simplicity, we only consider liquidating (or selling) the asset.

Previously, different methods \cite{nevmyvaka2006reinforcement,lin2020deep,fang2021universal} are proposed to apply RL to trade execution, but they follow a similar procedure: The agent learns based on interactions with a data-driven simulator. The dataset contains the information collected from the real market and is used to determine simulated transitions.
Therefore, it can be regarded as an offline RL setting. The observation of the agent can be divided into the market variable (e.g., the LOB snapshot) and the private variable (e.g., remaining time and inventory). The market variable is usually high-dimensional and incorporates different forms of information to represent the noisy and partially observable market.

\subsection{Offline RL with Dynamic Context}
To model the problem structure when applying RL to trade execution, we introduce the ORDC model and show the diagram in Figure~\ref{fig:ORDC}. 
ORDC is a tuple $(\mathcal{X}, \mathcal{C}, \mathcal{S}, \phi, P, r, \gamma, \mathcal{D})$ specifying the latent context space $\mathcal{X}$, the context space $\mathcal{C}$, the state space $\mathcal{S}$, 
the unknown context decoding function $\phi:\mathcal{C}\to\mathcal{X}$,
the transition dynamics $P(x',s'|x,s,a)=P_x(x'|x)P_s(s'|x,s,a)$, the reward $r(x,s,a)$, the discount factor $\gamma$ and the offline context dataset $\mathcal{D}$ that contains context sequences.

In trade execution, the market variable serves as the context, and the private variable serves as the state.
The context evolves independently and is not affected by the action or the state. However, the context is important since it influences the reward collected by the agent and the dynamics of the state. The context $c\in\mathcal{C}$ is usually high-dimensional and corresponds to a more compact latent context $x\in\mathcal{X}$ (e.g., the key information for making trading decisions).
Specifically, they have the \emph{block structure} \cite{du2019provably}:
Each context $c\in\mathcal{C}$ uniquely determines its generating latent context $x\in\mathcal{X}$ with the unknown mapping $\phi:\mathcal{C}\to\mathcal{X}$.

\begin{figure}[t]
\centering
  \includegraphics[width=\linewidth]{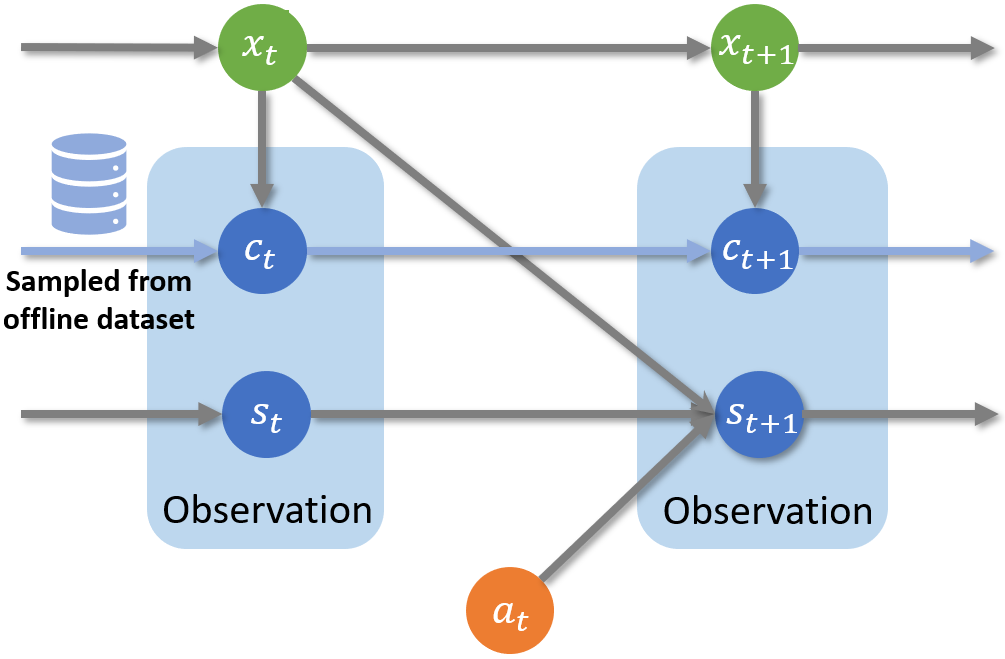}
\caption{
    Diagram of Offline RL with Dynamic Context (ORDC). The gray arrows indicate ``generates'' and the blue arrows indicate ``sampling from offline dataset''. }
\label{fig:ORDC}
\end{figure}

Given a policy $\pi:\mathcal{C}\times\mathcal{S}\to\Delta^{\mathcal{A}}$,
the Q-function is defined as $Q^\pi(c,s,a) := \mathbb{E} [\sum_{t=0}^\infty \gamma^t r(x_t, s_t, a_t)|x_0=\phi(c),s_0=s,a_0=a,\pi]$,
where $a_t \sim \pi(\cdot|c_t,s_t)$ 
and $x_t,s_t$ transits following the dynamics $P$
for all $t\ge 1$.
The agent learns from interactions with a data-driven simulator based on the offline dataset $\mathcal{D}$ and outputs a policy $\pi$ that maximizes
$J(\pi) = \mathbb{E}[Q^\pi(c,s,a)|(c,s)\sim P_0, a\sim \pi(\cdot|c,s)]$
where $P_0$ is the initial context/state probability.

\begin{figure}[t]
\centering
  \includegraphics[width=\linewidth]{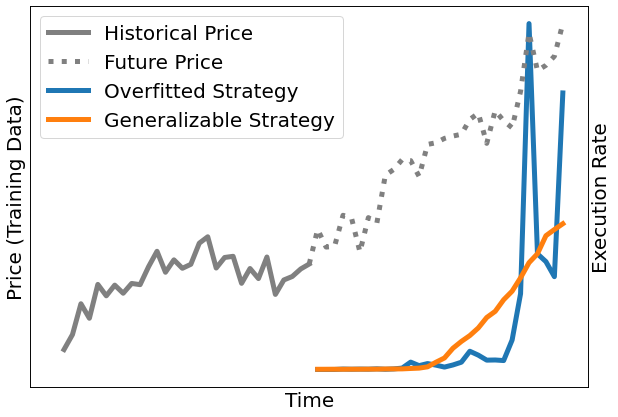}
\caption{
    The blue line: In trade execution, the deep learning agent can memorize and overfit to the context sequence used for training (the gray line) and liquidate most of the inventory on the highest price.
    The orange line: A generalizable agent should output a smooth policy considering the stochasticity of future context.
    See the corresponding experiment setting in Section \ref{sec:toy_example}. }
\label{fig:merge}
\end{figure}

\subsection{Theoretical Analysis for Generalization under ORDC}
With a slight abuse of notation, we can write the dynamics as $P(c',s'|c,s,a)=P_c(c'|c) P_s(s'|c,s,a)$ owning to the block structure.
In many real instances of ORDC, simulation is cheap and thus $P_s$ can be accurately estimated from a large number of interactions with the simulator.
Moreover, $P_s$ (e.g., the rules to match the orders) is usually simple, whereas $P_c$ (e.g., involving market dynamics)
is complex and hard to estimate.
Therefore, we further assume $P_s$ is known. The following sample complexity lower bound highlights the intrinsic difficulty in generalization under ORDC. 
(Notice that the sample complexity indicates how many samples are sufficient to ensure a small generalization gap.)

\begin{assumption}[Regularity]
\label{assumption:MDP_regularity}
\textit{
$\mathcal{C}$, $\mathcal{X}$, $\mathcal{S}$, and $\mathcal{A}$ are discrete and the immediate reward $r(x,s,a)\in[0,1], \forall x,s,a$.
}
\end{assumption}

\begin{theorem}
\label{theorem:lower_bound}
\textit{
Under Assumption \ref{assumption:MDP_regularity}, there exists a class of ORDC models $\mathbb{M}=\{ M_1, \cdots, M_m\}$ such that any algorithm $A$ needs at least $T=\Omega\Big( \frac{|\mathcal{C}|\log({|\mathcal{C}|}/{\delta})}{(1-\gamma)^3 \epsilon^2}\Big)$ context samples to learn a value function $Q^A \in \mathbb{R}^{\mathcal{C}\times\mathcal{S}\times\mathcal{A}}$ such that $\| Q^* - Q^A \|_\infty \le \epsilon$ with probability at least $1-\delta$ for all $M\in\mathbb{M}$, where $Q^*$ is the optimal action value function.
}
\end{theorem}

We provide the proof in appendix. In the proof, we construct a class of ORDC models where the contexts can be divided into a small number of categories.
However, without further knowledge on how to aggregate the contexts, the algorithm still needs a large number of samples (i.e., $\tilde{O}(|\mathcal{C}|)$) to learn a generalizable policy. We will later show that the sample complexity can be improved when the context aggregation is known.

The theorem indicates that the estimated value function can overfit to limited context sequences when the context space is large or the underlying context dynamics is complex. 
Actually, this is the case for real trade execution tasks.
First, the context space is large since people usually incorporate many market indicators as the context to reflect the market more comprehensively.
Second, the context dynamics is complex and highly stochastic since it is driven by various market participants, news, economics, etc.
Moreover, we find that function approximation does not effectively improve generalization since deep learning models can also suffer from such overfitting.
To illustrate this, we present an overfitted strategy and a generalizable strategy in Figure~\ref{fig:merge}.
The overfitted strategy results from a deep learning model trained using a standard RL algorithm, and the generalizable strategy results from a similar training but with a technique that aggregates the context (see Section \ref{sec:toy_example} for details).
We can see that a standard deep RL model can memorize the highest price in the training context (price) sequence and learn an aggressive policy that liquidates nearly all the stocks on that price.

\section{Towards Better Generalization}
\label{sec:better_generalization}

Motivated by the theoretical analysis, we first show that aggregating the context can improve the generalization theoretically. Then, we introduce two practical algorithms for trade execution: 1) \textbf{CASH} (\textbf{C}ontext \textbf{A}ggregate with \textbf{H}and-crafted \textbf{S}tatistics) , and 2) \textbf{CATE}(\textbf{C}ontext \textbf{A}ggregate with \textbf{E}nd-to-end \textbf{T}raining).

\subsection{Context Aggregation}
In ORDC, the agent may overfit to limited context sequences in the dataset.
We observe that, by resorting to the context decoding function $\phi:\mathcal{C} \to \mathcal{X}$ that maps the high-dimensional context into the latent context, we can obtain a more generalizable agent. 

\begin{theorem}
\label{theorem:upper_bound}
\textit{
With the access to $\phi$ and a generative model to collect context samples for $\mathcal{D}$, there exists an algorithm that learns a value function $\hat{Q}$ such that $\| Q^* - \hat{Q} \|_\infty \le \epsilon$ as long as the context transitions in $|\mathcal{D}|$ is larger than $D=O\Big( \frac{|\mathcal{X}|^2 \log(|\mathcal{X}|/\delta)}{(1-\gamma)^4 \epsilon^2} \Big)$ with probability at least $1-\delta$.
}
\end{theorem}

We provide the proof of Theorem ~\ref{theorem:upper_bound} in appendix.
\footnote{Compared with Theorem \ref{theorem:lower_bound}, the additional dependency on the cardinality of the (latent) context space may be improved using more involved analysis (e.g., using a Bernstein style inequality).}
Notice that $\mathcal{X}$ is considered to have a much smaller cardinality than $|\mathcal{C}|$.
This indicates that we can learn a good policy with much fewer context samples when $\phi$ is available.
However, the mapping $\phi$ is not provided in many real scenarios.
Nevertheless, we still hope to improve the generalization of the model by finding a mapping that can effectively aggregate the high-dimensional context.
Next, we propose two algorithms to approximate the mapping either using the domain knowledge or an end-to-end training scheme.

\subsection{Practical Algorithms}
With a simulator that replays the historical context sequences, we can use the standard online RL algorithm by treating it as a regular MDP (where the observation contains the context as well as the state).
Additionally, we consider the risk of overfitting highlighted in the analysis on ORDC and propose to train a context encoder to approximate the mapping $\phi$.

In our algorithms, we guide the learning of the context encoder with the statistics extracted from future contexts.
The reasons are as follows:
The key challenge in trade execution is to handle the uncertainty of the future since the solution for the task would be easy if we knew the future, e.g., liquidating on the known highest price.
In this sense, compared with traditional methods that do not use context and output static policies, using context gives us a good indication of the future and enables dynamic adaptation.
Therefore, the information extracted from future contexts is important to guide the training of the context encoder.
However, there is spurious noise in the future contexts that is not predictable from the current context.
Therefore, we hope to use stable statistics (i.e., with a stable correlation with the optimal decision, see \cite{arjovsky2019invariant}) that can extract predictable as well as generalizable information from future contexts. 
For example, crude price movement and volatility in the future are rather predictable, whereas the specific time point in the future when the price is highest is not generalizable.
Moreover, the design of our algorithms is also motivated by the experiments on a simplified trade execution task (See Section \ref{sec:toy_example}), which indicates that guiding the context encoder with the statistics on future contexts is a simple yet effective method.

\begin{figure*}[t]
	\begin{centering}
		\includegraphics[width=1.2\columnwidth]{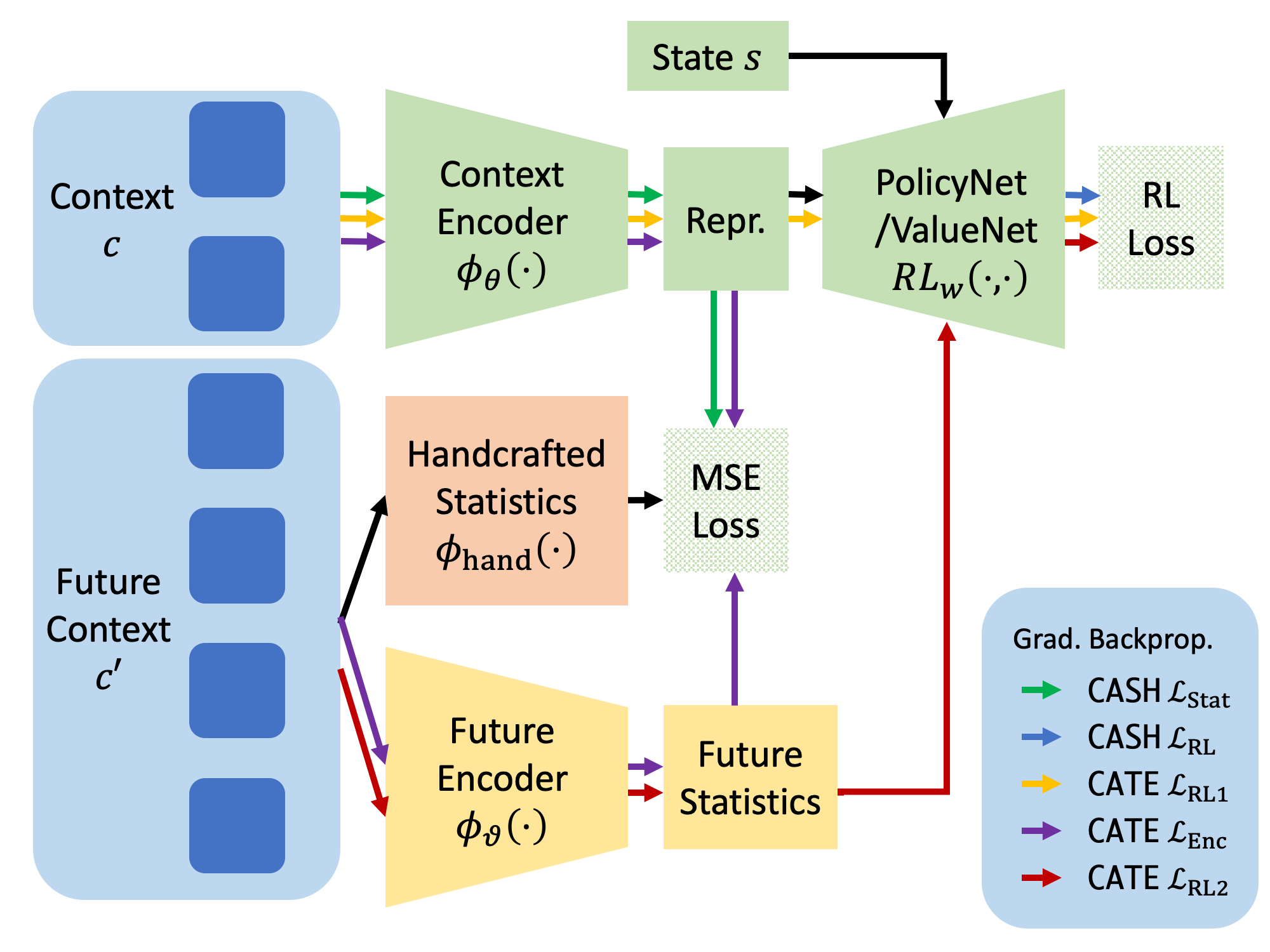}
		\caption{The architecture of variational autoencoder}
		\label{fig:algorithm}
	\end{centering}
\end{figure*}

\subsection{CASH: Context Aggregation with Hand-crafted Statistics}
CASH and CATE are based on the training of a standard DRL model $RL_w(\phi_\theta(c), s)$ that receives a context $c$ and a state $s$ and outputs actions or values. We use $\theta$ and $w$ to denote the parameters in the context encoder and the policy/value network respectively.
CASH uses hand-crafted statistics to guide the learning of the context encoder. We present the diagram of CASH in Figure~\ref{fig:algorithm}.
In the pre-training phase, we first train a context encoder that tries to predict the labels extracted using hand-crafted statistics from future contexts.
Specifically, we train the encoder $\phi_\theta(\cdot)$ with the loss $\mathcal{L}_\text{Stat}(\theta):= (\phi_\theta(c)-\phi_\text{hand}(c'))^2$ where $c'$ is the future context following the current context $c$ and $\phi_\text{hand}$ extracts hand-crafted statistics.
The loss is a mean-squared error between the context representation and the generated statistics.
In the training phase, we fix the context encoder $\phi_\theta(\cdot)$ and train the policy $RL_w(\cdot, \cdot)$ using standard RL algorithms.
We denote the loss function of the RL algorithm as $\mathcal{L}_\text{RL}(w)$ which can be the TD error for value-based RL algorithms (e.g., DQN) or the negative policy performance estimation in policy gradient algorithms (e.g., PPO).
These losses are estimated and optimized based on the transition samples collected from the data-driven simulator.
Specifically, the context $c$ and future context $c'$ are collected by replaying the historical data; the state $s$ and the reward are calculated by the simulator.

In our experiment, we design statistics of the future contexts {$\phi_\text{hand}(\cdot)$} based on the observation that the task would be simple if the information about the future price trend and spread is discovered.
We use the following hand-crafted statistics:
1) The difference between the average future twap (time-weighted average price) and the current twap, which indicates the trend;
2) The difference between the maximum/minimum future twap and the current twap, which indicates whether the current price is a peak/bottom;
3) The volatility of the future twap, which is related to risk control;
4) The standard deviation of future spreads and the difference between the average/maximum/minimum of future spreads, which are related to the temporary market impact.

\subsection{CATE: Context Aggregation with End-to-end Training}
Designing effective statistics for the context encoder to predict requires expertise in the specific domain, which is unavailable in many real scenarios.
Therefore, we propose CATE that learns to generate such statistics via a future (context) encoder $\phi_\vartheta(\cdot)$ where $\vartheta$ is the trainable parameter.
We present the algorithm in Figure \ref{fig:algorithm}.
The algorithm learns a context encoder (that outputs the context representation), a future encoder (that outputs the future statistics), and a policy/value network.
These components are trained simultaneously with the following loss in an in an end-to-end manner:
\begin{equation} \small \nonumber
\begin{aligned}
\mathcal{L}(\theta, \vartheta, w) & = \sum_{\text{transitions}} \mathcal{L}_\text{Enc}(\theta, \vartheta) + \mathcal{L}_\text{RL1}(\theta, w) + \mathcal{L}_\text{RL2}(\vartheta, w) \\
\text{where } & \mathcal{L}_\text{Enc}(\theta, \vartheta):= (\phi_\theta(c)-\phi_\vartheta(c'))^2 \\
& \mathcal{L}_\text{RL1}(\theta, w) \text{ optimizes the RL model } {RL}_w(\phi_\theta(c), s) \\
& \mathcal{L}_\text{RL2}(\vartheta, w) \text{ optimizes the RL model } {RL}_w(\phi_\vartheta(c'), s)
\end{aligned}
\label{eq:CATE}
\end{equation}
The first term $\mathcal{L}_\text{Enc}$ is used not only to train the context encoder but also encourage the future encoder to generate predictable statistics.
The second term $\mathcal{L}_\text{RL1}$ is the RL loss that optimizes the base RL model (i.e., the context encoder and the value/policy network).
The third term $\mathcal{L}_\text{RL2}$ encourages the future encoder to generate statistics that are helpful for decision making and therefore incentivizes informative future statistics.
Similar to CASH, this loss is calculated based on the transition samples collected from the data-driven simulator.

\section{Experiments}
\label{sec:execution_experiment}

\begin{table*}[t]
    \centering

\begin{tabular}{llrrrrrr}
\toprule
   &                     & \multicolumn{2}{c}{Reward (train)} & \multicolumn{2}{c}{Reward (eval)} & \multicolumn{2}{c}{Gap} \\
   Data volume &  Model  &              mean &    std &             mean &    std &    mean &    std \\
\midrule
100k & \emph{Base} &            1.8559 & 0.0379 & \textbf{1.8291} & 0.0227 & 0.0268 & 0.0317 \\
   & \emph{Bottleneck} &           -0.0024 & 0.0035 &          -0.0022 & 0.0162 & -0.0003 & 0.0150 \\
   & \emph{CATE} &            1.7143 & 0.0281 &           1.7004 & 0.0265 &  0.0139 & 0.0252 \\
   & \emph{CASH} &            1.8016 & 0.0246 &           1.8028 & 0.0187 & -0.0013 & 0.0321 \\
\midrule
10k & \emph{Base} &            2.5389 & 0.0197 &           1.6047 & 0.0266 &  0.9341 & 0.0346 \\
   & \emph{Bottleneck} &            0.0059 & 0.0131 &          -0.0006 & 0.0013 &  0.0065 & 0.0144 \\
   & \emph{CATE} &            1.7743 & 0.0261 &           1.7335 & 0.0536 &  0.0408 & 0.0471 \\
   & \emph{CASH} &            1.8329 & 0.0300 & \textbf{1.8083} & 0.0304 &  0.0246 & 0.0544 \\
\midrule
1k & \emph{Base} &            4.0340 & 0.0659 & 1.4083 & 0.0748 &  2.6258 & 0.1259 \\
   & \emph{Bottleneck} &            0.3202 & 0.8233 &           0.4142 & 0.8833 & -0.0940 & 0.1430 \\
   & \emph{CATE} &            2.0324 & 0.1823 &           1.6007 & 0.0591 &  0.4317 & 0.1736 \\
   & \emph{CASH} &            2.1578 & 0.1219 & \textbf{1.9557} & 0.0555 &  0.2021 & 0.1188 \\
\bottomrule
\end{tabular}

    \caption{The performance of different models on the simplified trade execution task. The models are evaluated over five random seeds. 
    Reward (train/eval) represents the negative trading cost on the training/evaluation set.
    }
    \label{tab:toy}
\end{table*}

In this section, we first conduct a simplified task of trade execution, to illustrate the overfitting in vanilla RL methods and the effectiveness of context aggregation for generalization. Then, based on our market simulator, we compare our algorithms with other existing trade execution methods in the real stock market data.
We provide the source code in \href{https://github.com/zhangchuheng123/RL4Execution}{https://github.com/zhangchuheng123/RL4Execution}.

\subsection{Experiments on Simplified Trade Execution Task}
\label{sec:toy_example}

In this experiment, we introduce a simplified trade execution task to study the overfitting problem with a context dataset and possible solutions towards better generalization.
In this simplified task, all transactions are executed on a single price process (i.e., without ask/bid price) following the Brownian motion. Formally, the dynamic of price process is $p_{t+1} = p_t + \Delta p_t, \Delta p_t = \alpha + \sigma \xi_t$, where $p_t$ is the price at the $t$-th step, $\xi_t \sim \mathcal{N}(0,1)$ is a random variable which follows the the standard Gaussian distribution independently at each step $t$, $\alpha$ and $\sigma$ denote drift and volatility which are two statistical parameters of process. 
The task is to learn an agent that can give a  execution strategy based on the observation of price changes over the past 30 time steps.

The objective is to reduce the trading cost, and therefore we set the negative discounted trading cost as the reward (i.e., the gap between the average discounted execution price and the average discounted market price).


This task is a simplified ORDC task, which focuses on 1) the existence of a mapping between the high-dimensional context $c$ and the latent context $x=(\alpha, \sigma)$, and 2) training with a limited context dataset.
In the following experiments, we use DDPG \cite{silver2014deterministic} as the base RL algorithm.
We design several methods to solve the task and observe the corresponding generalization ability of the learned agents.
We present the experiment results in Table \ref{tab:toy} and analyze the result of each model as follows:

\textbf{\emph{Base.}} 
We can first observe the performance of the base RL model. We can see that it performs well when the sample data is sufficient, but its performance degenerates quickly when the data volume decreases.

\textbf{\emph{Bottleneck.}} 
When the data is noisy and limited, deep learning models with high capacity are able to memorize the samples in the training set.
A natural idea is to limit the model capacity with a representation bottleneck (i.e., learning a low-dimensional representation).
With the prior knowledge that the whole price process can be represented by two parameters (i.e., drift $\alpha$ and volatility $\sigma$), so we set the representation to be a two-dimensional vector.
However, we observe that an end-to-end training process for a model with bottleneck does not result in a good performance.
Additionally, we observe that the training process is highly unstable.

\textbf{\emph{CATE.}}
In this model, we consider an encoder-decoder architecture to learn the representation with the others remaining the same as \emph{Bottleneck}.
Specifically, the encoder generates a two-dimensional representation from the past context, and the decoder tries to predict the future context sequence based on the representation (the decoder that reconstructs the past context sequence results in similar performance).
The decoder receives a two-dimensional representation generated by the encoder and tries to predict the future context sequence.
We observe that this model results in relatively good performance even when the data is highly limited.
However, the two-dimensional representation may not only embed the estimated statistical parameters (i.e., $(\sigma, \alpha)$) but also overfit the spurious noise in the data (i.e., $\xi_t$s).

\textbf{\emph{CASH.}}
In this model, we consider using a separate training signal to supervise the learning of the representation with other architectures remaining the same as \emph{Bottleneck}.
The loss function to train the encoder is a mean-squared-error w.r.t. a two-dimensional hand-crafted target vector, which is the estimate of  $(\hat{\sigma}, \hat{\alpha})$.
We observe that this model achieves superior performance even when the data is highly scarce. Moreover, the generalization gap is only half of that in the previous model, which may benefit from the fact that this model avoids fitting the spurious noise.

We also present the strategies learned by \emph{Base}  and \emph{CASH} with 1k data in Figure \ref{fig:merge} (right) and more figures in appendix.
We can observe that the strategy learned by \emph{Base} (cf. overfitted strategy in Figure \ref{fig:merge} right) presents sharp peaks resulting from overfitting the training data.
In contrast, the strategy learned by \emph{CASH} (cf. generalizable strategy in Figure \ref{fig:merge} right) is smooth which indicates that the agent is more generalizable.

\textbf{Conclusion.}
Through the experiments on this simplified trade execution task, we have several observations: 1) Overfitting can easily occur for a deep RL model even in a setting simpler than the ORDC model. 2) Simply regularizing the capacity of the representation does not lead to better performance or generalization. 3) Reconstruction/Prediction-based encoder training combined with limited representation capacity can achieve good performance.
4) With carefully designed target features, we can prevent the encoder from fitting spurious noise and further improve generalization.

\subsection{Experiments with High-Fidelity Simulation}
\begin{table*}[t]
    \centering
        \begin{tabular}{lrrrc}
        \toprule
        Algorithm & Training & Validation & Testing & Gap \\
        \midrule
        TWAP  & - & - & 14.0984 (2.1545) & - \\
        Momentum & - & - & 12.2530 (0.6151) & - \\
        \midrule
        Tuned DQN & 2.0382 (1.7684) & 5.8134 (2.1032) & 5.9240 (3.2986) & 3.8858 \\
        \quad \cite{nevmyvaka2006reinforcement} & 3.0781 (5.2447) & 8.8698 (1.5701) & 9.1223 (1.0554) & 6.0441 \\
        \quad \cite{ning2018double} & 7.3248 (5.1059) & 10.3971 (2.0066) & 9.4051 (2.6524) & 2.0804 \\
        \quad \cite{lin2020deep} & 5.8778 (7.0791) & 10.7000 (0.7024) & 12.7116 (1.1514) & 6.8338 \\
        \quad Tuned DQN + {CASH} & 2.2269 (2.0798) & 3.7992 (1.4612) & 3.4250 (1.9052) & \textbf{1.1981} \\
        \quad Tuned DQN + {CATE} & -2.8774 (1.7019) & -1.8431 (1.2983) & \textbf{0.0075} (1.6920) & 2.8849 \\
        \midrule
        Tuned PPO & -0.9505 (2.2439) & 2.1132 (0.2497) & 2.7575 (1.2070) & 3.7079 \\
        \quad \cite{daberius2019deep} & 7.6944 (11.3490) & 9.2893 (1.9000) & 12.3166 (2.8627) & 4.6222\\
        \quad \cite{lin2020end} & 5.2697 (7.4173) & 8.1153 (0.8894) & 9.4807 (1.9686) & 4.2110 \\
        \quad \cite{fang2021universal} & -4.9090 (16.1474) & 10.1338 (4.0843) & 11.8739 (5.0948) & 16.7829 \\
        \quad Tuned PPO + {CASH} & -4.6504 (0.4916) & -3.9351 (0.2810)  & -4.5760 (0.2062) & \textbf{0.0744} \\
        \quad Tuned PPO + {CATE} & -5.0364 (0.8104) & -3.8797 (0.1103) & \textbf{-4.9068} (0.3015) & 0.1296 \\
        \bottomrule
        \end{tabular}
        
    \caption{The trading cost (bp=$10^{-4}$) of different algorithms. 
    The validation set is used for hyperparameter tuning.
    The numbers are the average mean (std.) trading cost in the last 100 evaluations of the total 1000 evaluations over five different random seeds. The bold numbers indicate the algorithms with the best performance or smallest generalization gap.}
    \label{tab:main_exp}
\end{table*}





\textbf{Simulated Environment.}
To reduce the gap between simulation and the real-world environment and provide a reliable benchmark for different algorithms, we build an open source, high fidelity, and flexible market simulator for trade execution.
Compared with previous simulators that are based on bar-level simulation \cite{fang2021universal}, our simulator is based on LOB-level data and thus has higher fidelity.
Specifically, our simulator considers the temporary market impact, time delay, and second-level snapshot-by-snapshot order book reconstruction 
to minimize the sim-to-real gap.
Moreover, our simulator can easily adapt for different designs (e.g., in the action/observation space and the reward function) of previous methods and therefore enables comparing different methods uniformly.
See detailed description for the simulator and the environment settings in appendix.

\textbf{Experiment Setting.}
Our simulator is based on the LOB data of 100 most liquid stocks in China A-share market. The data collected from April 2022 to June 2022 is used as the training set, and the data collected during July 2022 and August 2022 are used as the validation and testing set respectively.
The task is to sell $0.5\%$ of the total trading volume of the last trading day in a 30-minute period randomly selected from a trading day.
The agent makes a decision (i.e., placing orders) at the start of each minute.
We train an universal model for all the stocks.
The evaluation metric is the trading cost defined as $(\bar{p}_\text{TWAP} - \bar{p}) / \bar{p}_\text{TWAP}$, 
where $\bar{p}=A_{a}/V_{a}$ is the average execution price of the agent and $A_{a}, V_{a}$ are the trading money and volume of the agent respectively, $\bar{p}_\text{TWAP}=\frac{1}{T}\sum_{t=1}^{T} p_t$ is the time-weighted average price in the given $T$ time steps.
Trading cost is measured in basis point (bp) which is $10^{-4}$.

\textbf{Baselines.} 
We compare our algorithm with some rule-based and RL-based strategies for trade execution, where \emph{TWAP} divides a large order into smaller orders of equal quantities and executing them at regular intervals throughout the entire period, \emph{Momentum} buys relatively more quantity when the price rises, and vice versa. 

\begin{table}[ht]
    \centering
    \begin{tabular}{lrrrc}
    \toprule
    Predicted &   Price &  Volume \\
    Statistics &    &   \\
    \midrule
    Avg future twap - current twap & 4.5540  & 0.0971  \\
    Max future twap - current twap & 1.5008  & -0.2111 \\
    Min future twap - current twap & -1.0073 & 0.3853  \\
    Twap volatility                & 1.3095  & -0.4248 \\
    Avg future sprd - current sprd & -0.1326 & 0.1083  \\
    Max future sprd - current sprd & 1.8983  & -0.1352 \\
    Min future sprd - current sprd & -1.3171 & 0.0292  \\
    Sprd volatility                & 6.3998  & -0.8049 \\
    \bottomrule
    \end{tabular}
    \captionof{table}{The impact of predicted statistics on the agent's action in CASH (based on PPO).}
    \label{tab:interpretability}
\end{table}

\begin{figure*}[th]
    \centering
    
    \includegraphics[width=1.5\columnwidth]{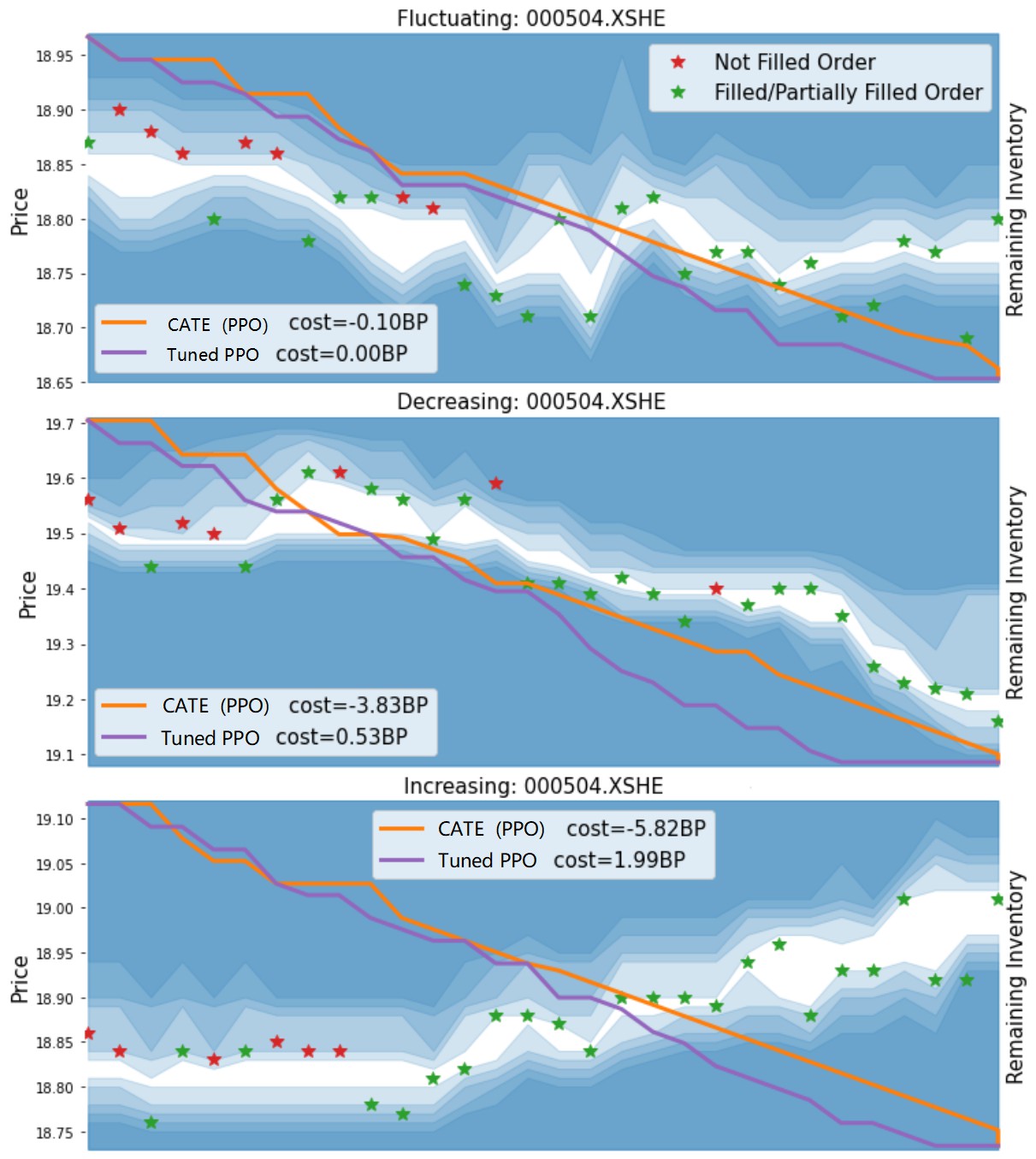}
    \caption{
    The policies under different trends learned by CATE (based on PPO) and the corresponding baseline (Tuned PPO). 
        The background shaded areas indicate the 5-level ask/bid prices, and
        the stars indicate the orders placed by {CATE}.
        The lines indicate the remaining inventory of {CATE} and the baseline algorithm.
    }
    \label{fig:example}
\end{figure*}

\textbf{Results.}
The experiment results are shown in Table \ref{tab:main_exp}.
We implement two families of algorithms based on DQN (which represents value-based RL methods) and PPO (which represents policy-based RL methods) respectively.
Most of the previous RL-based trade execution algorithms are based on these two base RL algorithms.
We implement these algorithms by following their designs in the model architecture, the observation, the action space, the reward function, etc. 
Moreover, we conduct experiments on the combinations of different designs and develop two well-designed RL-based trade execution algorithms (i.e., tuned DQN and tuned PPO in the table). Later, we implement our algorithms based on these two baselines.

First, we observe that tuned DQN/PPO outperforms other DQN/PPO-based baselines due to better designs.
Specifically, we found that using only MOs for trade execution leads to high trading costs to cross the ask-bid spread.
Moreover, the design of the reward function has a significant impact on the performance of the model.
The designs of tuned DQN/PPO and the impact of different designs can be found in appendix.
Second, we observe that 
our algorithms outperform tuned DQN/PPO due to aggregating the context.
Notice that OPD \cite{fang2021universal} uses a teacher policy that is based on the perfect information to guide the learning of the target policy, which is similar to our algorithms in extracting information from future contexts.
However, OPD suffers from a larger generalization gap since the guidance of the teacher policy is informative but may not be generalizable.
In contrast, the guidance in our algorithms (hand-crafted or generated statistics) is designed to be both informative and generalizable. 

The advantage of using hand-crafted statistics to guide the learning of the context encoder in CASH is that the meaningful context representations can lead to an interpretable learned policy.
To interpret the learned policy,
we estimate how each dimension of the representation (i.e., the predicted statistics) affects the selected quoted price and volume based on a set of collected context representations.
We list the slopes estimated using linear regression in Table \ref{tab:interpretability}.
In this way, we can examine the learned policy.
For example, we can see from the first row that the quoted price is 4.554 bp higher for every unit increase in the predicted price trend.

To present the policy of the agent under different trends, we plot an example of trade execution, as shown in Figure \ref{fig:example}.
Roughly speaking, the agent liquidates evenly across the horizon.
Meanwhile, the agent can adaptively place the order according to the trend and the timeline.
For example, the agent tends to place LOs at higher price levels in the early stage of the horizon and becomes more conservative in the latter stage.
Moreover, when there is a rising trend, the agent is more inclined to quote at a higher price to catch the trend.
We also show how the baseline PPO algorithm liquidates (cf. the purple lines).
Compared with our algorithm, the baseline tends to liquidate more at some specific steps which may result from overfitting the training data.
and the baseline tends to complete liquidation before the end of the given horizon which may lose trading opportunities.

\section{Conclusion}
\label{sec:execution_conclusion}
To analyze the overfitting problem when applying RL to the trade execution task, 
we propose an Offline RL with Dynamic Context (ORDC) framework. 
In this framework, we derive the generalization bound for the ORDC and find that the generalization gap results from limited data and large context space.
Motivated by the theoretical analysis, we propose to aggregate the context space to learn a generalizable agent.
Accordingly, we design two algorithms: \emph{CASH} that learns an interpretable agent using hand-crafted future statistics to guide context aggregation, and \emph{CATE} that learns a compact context representation without resorting to domain knowledge.
The experiments on both a simplified trade execution task and a well-designed high-fidelity simulated environment show that our algorithms can generate more generalizable agents.
Moreover, combined with a better design on the model components (e.g., the reward function and the action space), our algorithms achieve significant performance improvement over the previous methods.
In the future, we plan to apply the ORDC framework into other real-world RL applications that learn from offline context data.


\clearpage
\textbf{\Huge Appendix}
\appendix
\section{Limit Order Book}
\label{app:lob}
We present an LOB snapshot of one stock in Table \ref{table:LOB} which consists of five levels of the ask/bid prices and volumes.
\emph{Ask/bid prices and volumes} indicate that there are specific volumes of stock to sell/buy at the specified prices.
The first level of ask/bid price (i.e., the lowest ask price or the highest bid price) is referred to as \emph{the best ask/bid price}.
The \emph{mid price} is the average of the best ask price and the best bid price, and the \emph{spread} is the gap between them.
For example, the mid price is $(\$29.11 + \$29.01) / 2 = \$29.06$ and the spread is $\$29.11-\$29.01=\$0.10$ on the given snapshot.
The traders can trade via two types of orders: market orders (MOs) and limit orders (LOs).
An MO specifies the volume and is executed immediately with the best available price.
For example, an MO that sells $500$ shares of the stock will be executed at the average price $(\$29.01 \times 100 + \$29.00 \times 300 + \$28.99 \times 100) / 500 = \$29.00$. 
We observe that the average execution price is lower than the mid price, and the gap $\$29.06 - \$29.00=\$0.06$ is referred to as the \emph{temporary market impact}.
An LO specifies the volume as well as the price such that the trader will buy/sell the asset with a price no higher/lower than the preset price.
If LO is not executed immediately, it will be left in the LOB and appended to the order queue on the corresponding price level. 

\begin{table}[tbhp]
\centering
\begin{tabular}{ crr  }
 & Price & Volume \\
 \hline
 Ask 5 & \$29.15 & 10,000 \\
 Ask 4 & \$29.14 & 2,000 \\
 Ask 3 & \$29.13 & 1,000 \\
 Ask 2 & \$29.12 & 100 \\
 Ask 1 & \$29.11 & 200 \\
 \hline
 Bid 1 & \$29.01 & 100 \\
 Bid 2 & \$29.00 & 300 \\
 Bid 3 & \$28.99 & 800 \\
 Bid 4 & \$29.95 & 1,100 \\
 Bid 5 & \$29.09 & 1,900 \\
 \hline
\end{tabular}
\vspace{0.5cm}
\caption{A snapshot of the limit order book.}
\label{table:LOB}
\end{table}

\section{Proof of Theorem 1}
\label{app:proof1}
The proof of the lower bound mainly follows \cite{azar2013minimax} which studies the standard MDP.
First, we construct a class of ORDC models $\mathbb{M}$.
Then, following the analysis in \cite{azar2013minimax}, we can obtain the result that any algorithm in a subset of algorithms $A \in \mathfrak{A}' \subset \mathfrak{A}$ can fail to learn an accurate value function for at least one of the models in $\mathbb{M}$ with a high probability if the context dataset is not large enough.
At last, we generalize this conclusion to any of the possible algorithms $A\in\mathfrak{A}$ by showing that we can always find an algorithm $A'\in\mathfrak{A}'$ that performs no worse than $A$.

\begin{proof}
We first define a class of ORDC models $\mathbb{M}$.
For each instance $M$ in the class, the reward depends only on the context, i.e., $r(c,s,a)=r(c)$.
The context space is divided into three disjoint subsets of equal cardinality, i.e., $\mathcal{C} = \mathcal{C}^0 \cup \mathcal{C}^1 \cup \mathcal{C}^2$, $|\mathcal{C}|=3K$, and $|\mathcal{C}^0|=|\mathcal{C}^1|=|\mathcal{C}^2|=K$.
For each context in $c^0 \in \mathcal{C}^0$, it will trainsit to the corresponding context $c^1 \in \mathcal{C}^1$ with probability 1.
For each context $c^1 \in \mathcal{C}^1$, it will transit to the corresponding context $c^2 \in \mathcal{C}^2$ with probability $1-p_M$ and to itself with probability $p_M$.
For each context $c^2 \in \mathcal{C}^2$, it will transit to itself with probability 1.
The transition probability $p_M$ for each $c\in\mathcal{C}^1$ is selected from $\{p, p+\alpha\}$, where $\alpha$ and $p$ satisfy $0<p<p+\alpha<1$, and the exact value is determined in the analysis.
The reward function is 
\begin{equation}
    r(c) = \begin{cases}
    1 & \text{if } c\in \mathcal{C}^1 \\
    0 & \text{otherwise}
    \end{cases}.
\end{equation}
It is not hard to see that $Q_M^*(c,s,a)=\frac{\gamma}{1-\gamma p_M} =: \mathbb{Q}_M^*(c), \forall c\in\mathcal{C}^0$.

Now, we consider a subset of algorithms $\mathfrak{A}' \subset \mathfrak{A}$.
Each algorithm $A\in\mathfrak{A}'$ that consumes $T$ samples from $c^1\in\mathcal{C}^1$ outputs a value function that takes the same value for different $(s,a)$s, i.e., $Q_T^A(c,s,a) = Q_T^A(c,s',a') =: \mathbb{Q}_T^A(c), \forall s,a,s',a'$. 
Then, we can obtain the following conclusion by replacing $\mathcal{S}\times\mathcal{A}$ with $\mathcal{C}$ in Lemma 18 in \cite{azar2013minimax}:
\begin{lemma}
\textit{
For $\delta \in (0,1/2)$ and any algorithm $A\in\mathfrak{A}'$ using a total number of context transition samples less than $T=c_1 \frac{|\mathcal{C}|}{(1-\gamma)^3 \epsilon^2} \log ( c_2 |\mathcal{C}|/\delta )$, there exists $M_m\in\mathbb{M}$ such that 
$$
\mathbb{P}_m (\| \mathbb{Q}^*_{M_m} - \mathbb{Q}^A_T \|_\infty > \epsilon) > \delta,
$$
where $\mathbb{P}_m$ is the probability under the model $M_m$ and $c_1, c_2$ are positive constants.
}
\end{lemma}

At last, we extend the conclusion to $\mathfrak{A}$ with the following lemma.
\begin{lemma}
\textit{
For any algorithm $A\in \mathfrak{A}$ that outputs $Q^A \in \mathbb{R}^{\mathcal{C}\times\mathcal{S}\times\mathcal{A}}$, we can always find an algorithm $A'\in\mathfrak{A}'$ such that 
$$
\| Q^*_M - Q^A \|_\infty \ge \| \mathbb{Q}^*_M - \mathbb{Q}^{A'} \|_\infty
$$
}
\end{lemma}
Actually, we can construct an algorithm $A'\in\mathfrak{A}'$ for each $A\in\mathfrak{A}$ by wrapping the output of $A$ as follows: $\mathbb{Q}^{A'}(c) := Q^A(c,s_0, a_0)$ for arbitrary fixed $s_0\in\mathcal{S}, a_0\in\mathcal{A}$.
\end{proof}

\textbf{Discussion.}
In the class of constructed ORDC, we can categorize the contexts into six groups, three of which are the contexts $c^0\in\mathcal{C}^0$, $c^1\in\mathcal{C}^1$ and $c^2\in\mathcal{C}^2$ corresponding to $p_M=p$ and the other three corresponds to $p_M=p+\alpha$.
Aggregating the contexts into these six groups does not bring any loss in representing the optimal policy or value function.
In fact, it is possible to get rid of the $|\mathcal{C}|$ dependence in the sample complexity lower bound if the category of each context is known which is illustrated in Theorem 2.

\section{Proof of Theorem 2}
\label{app:proof2}
In the following proof, we use the lower case $p$ to represent the transition probability and the upper case $P$ to represent the corresponding matrix form, which is slightly different from the notation in the main text.
With access to $\phi$, the underlying transition model can be written as $p(x',s'|x,s,a) = p_x(x'|x) p_s(s'|x,s,a)$.
With a generative model, we can collect $N$ context transition samples from the context that corresponds to each latent context $x\in\mathcal{X}$.
Therefore, we consume a total number of $N|\mathcal{X}|$ samples.
Our algorithm estimates $p_x$ with these samples as follows: 
$\hat{p}_x(x'|x) = \text{count}(x,x') / \text{count}(x)$,
where $\text{count}(x,x')$ is the number of transitions from $x$ to $x'$ in the samples and $\text{count}(x)$ is the number of transitions starting from $x$ in the samples.

We denote $\hat{p}(x',s'|x,s,a) = \hat{p}_x(x'|x) p_s(s'|x,s,a)$ and the matrix form as $\hat{P}\in\mathbb{R}^{ |\mathcal{X}||\mathcal{S}||\mathcal{A}| \times |\mathcal{X}||\mathcal{S}| }$.
We can link the estimation error of $\hat{p}_x$ to that of $\hat{p}$ with the following lemma:
\begin{lemma}
\label{lemma:link}
\textit{
For any $x\in\mathcal{X}, s\in\mathcal{S}, a\in\mathcal{A}$, we have
$$
\| p(\cdot|x,s,a) - \hat{p}(\cdot|x,s,a) \|_1 = \| p_x(\cdot| x) - \hat{p}_x(\cdot|x) \|_1
$$
}
\end{lemma}
\begin{proof}[Proof of Lemma \ref{lemma:link}]
$$
\begin{aligned}
& \| p(\cdot|x,s,a) - \hat{p}(\cdot|x,s,a) \|_1 \\
= & \sum_{x',s'} |p(x',s'|x,s,a) - \hat{p}(x',s'|x,s,a)|  \\
= & \sum_{x',s'} |p_x(x'|x) p_s(s'|x,s,a) - \hat{p}_x(x'|x) p_s(s'|x,s,a) |  \\
= & \sum_{x',s'} |p_x(x'|x) - \hat{p}_x(x'|x) | p_s(s'|x,s,a)  \\
= & \Big( \sum_{x'} |p_x(x'|x) - \hat{p}_x(x'|x) | \Big) \Big( \sum_{s'} p_s(s'|x,s,a) \Big)  \\
= & \| p_x(\cdot| x) - \hat{p}_x(\cdot|x) \|_1
\end{aligned}
$$
\end{proof}

Next, we introduce several lemmas for MDP. 
Notice that they also apply to our setting by treating the state in MDP as $(x,s)$ in the ORDC model.
\begin{lemma}[Lemma 2.2 in \cite{agarwal2019reinforcement}]
\label{lemma:22}
$$
Q^\pi - \hat{Q}^\pi = \gamma (I - \gamma \hat{P}^\pi)^{-1} (P - \hat{P}) V^\pi, \forall \pi
$$
\end{lemma}
\begin{lemma}[Lemma 2.3 in \cite{agarwal2019reinforcement}]
\label{lemma:23}
\textit{
For any policy $\pi$, and vector $v\in\mathbb{R}^{|\mathcal{X}||\mathcal{S}||\mathcal{A}|}$, we have
$$
\| (I - \gamma P^\pi)^{-1} v \|_\infty \le \|v\|_\infty / (1-\gamma)
$$
}
\end{lemma}
\begin{lemma}[Lemma A.8 in \cite{agarwal2019reinforcement}]
\label{lemma:A8}
\textit{
Let $z$ be a discrete random variable that takes values in $\{1, \cdots, d\}$ distributed according to $q\in\Delta^d$ with $q_i = Pr(z=i)$. 
Assume we have $N$ iid samples, and the empirical estimate $\hat{q}\in\Delta^d$ with $\hat{q}_i = \sum_{i=1}^N {\bf 1}[z_i=i] / N$.
For some constant $c>0$ and w.p. at least $1-\delta$, we have
$$
\| q - \hat{q} \|_1 \le c \sqrt{\dfrac{d \log(1/\delta)}{N}}.
$$
}
\end{lemma}

\begin{proof}[Proof of Theorem 2]
For any policy $\pi$, we have
$$
\begin{aligned}
& \| Q^\pi - \hat{Q}^\pi \|_\infty \\
= & \| \gamma (I - \gamma \hat{P}^\pi)^{-1} (P - \hat{P}) V^\pi \|_\infty \\
\le & \frac{\gamma}{1-\gamma}  \| (P - \hat{P}) V^\pi \|_\infty \\
\le & \frac{\gamma}{1-\gamma} \Big( \max_{x,s,a} \| p(\cdot|x,s,a) - \hat{p} (\cdot|x,s,a) \|_1  \Big)  \| V^\pi \|_\infty \\
= & \frac{\gamma}{1-\gamma} \Big( \max_{x} \| p_x(\cdot|x) - \hat{p}_x (\cdot|x) \|_1  \Big)  \| V^\pi \|_\infty \\
\le & \frac{\gamma}{(1-\gamma)^2} \Big( \max_{x} \| p_x(\cdot|x) - \hat{p}_x (\cdot|x) \|_1  \Big) \\
\le & c \frac{\gamma}{(1-\gamma)^2} \sqrt{\dfrac{|\mathcal{X}| \log(|\mathcal{X}|/\delta)}{N}} \quad \text{w.p. at least } 1-\delta. \\
\end{aligned}
$$
The first line uses Lemma \ref{lemma:22}. 
The second line uses Lemma \ref{lemma:23}.
The third line uses Holder's inequality.
The fourth line uses Lemma \ref{lemma:link}.
The fifth line uses Assumption 1.
The sixth line uses Lemma \ref{lemma:A8} with a union bound.
Equivalently, with $|\mathcal{X}| N = \dfrac{c^2 \gamma^2}{(1-\gamma)^4} \dfrac{|\mathcal{X}|^2 \log(|\mathcal{X}|/\delta)}{\epsilon^2}$ samples, we have $\| Q^\pi - \hat{Q}^\pi \|_\infty \le \epsilon$ w.p. at least $1-\delta$.

Then, for all $(x,s,a)\in\mathcal{X}\times\mathcal{S}\times\mathcal{A}$, we have
$$
\begin{aligned}
& |Q^*(x,s,a) - \hat{Q}^*(x,s,a)| \\
= & |\sup_\pi Q^\pi(x,s,a) - \sup_\pi \hat{Q}^\pi(x,s,a) | \\
\le & \sup_\pi |Q^\pi(x,s,a) - \hat{Q}^\pi(x,s,a)| \\
\le & \epsilon,
\end{aligned}
$$
which completes the proof.
\end{proof}

\textbf{An illustrative example.}
In addition to the proof, we also provide an illustrative view on how context aggregation can help generalization.
We consider a policy evaluation setting where we aim to estimate the value function for some context.
We will later show that, by aggregating the context properly, we can obtain a more accurate estimate of the value.

\begin{figure*}
    \centering
    \includegraphics[width=1.4\columnwidth]{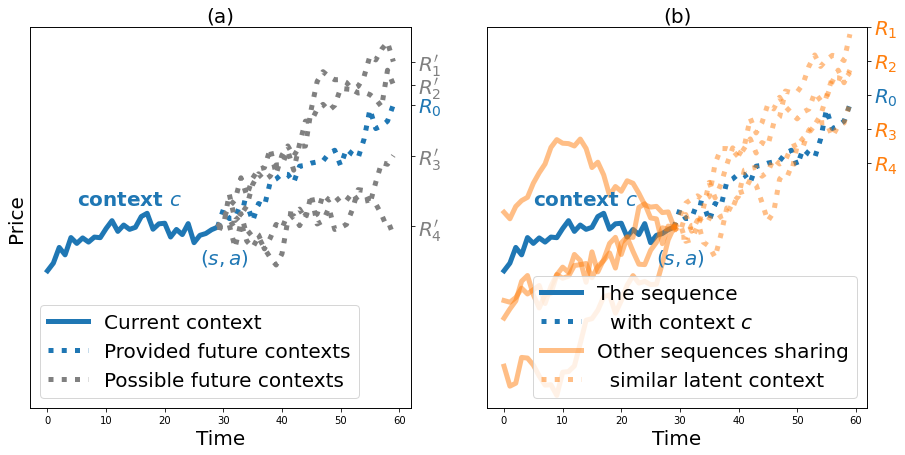}
    \caption{
    a) When a resettable environment is available, we can estimate $Q(c,s,a)$ accurately with the returns calculated based on fictitious context sequences (the gray lines).
    However, in real scenarios where fictitious context sequences cannot be obtained, we can only use the expected return calculated on few training sequences (the blue line) as the estimate which suffers from a large variance.
    b) When resorting to a mapping between the high-dimensional context and the latent context, we can obtain a better estimate by averaging over the returns calculated based on similar sequences (the orange lines). 
    }
    \label{fig:theorem2_illustration}
\end{figure*}

Let us consider to estimate a value function $Q(c,s,a)$ for some $(c, s, a)\in\mathcal{C}\times\mathcal{S}\times\mathcal{A}$.
We illustrate the procedure under the simplified trade execution task (described in Section 5.1) in Figure \ref{fig:theorem2_illustration}.

When it is possible to reset the environment to $(c,s)$ and then take the action $a$ to rerun unlimited number of times, an accurate estimation of $Q(c,s,a)$ is $\hat{Q}'(c,s,a) = \frac{1}{N} (R'_1 + R'_2 + R'_3 + R'_4 + \cdots)$, where $R'_i$ is the cumulative return from $(c,s,a)$ in the $i$-th run (cf. the gray dashed lines in Figure \ref{fig:theorem2_illustration}a) and $N$ is the total number of runs.
When the number of samples is sufficient, $\hat{Q}'(c,s,a)$ approaches the true value $Q(c,s,a)$.

However, a resettable environment is not provided in most realistic scenarios, and therefore it is hard to obtain $R'$s.
For example, in the trade execution task, it is hard to reset to a specific market context and perform a counterfactual simulation on the complex market dynamics.
Without context aggregation (via either implicit function approximation or explicit context encoder), a straightforward way to estimate $Q(c,s,a)$ is to use the sample mean.
Suppose the context $c$ only appears once in the dataset (which is quite common when the context is high-dimensional).
In this scenario, the best estimate should be $\hat{Q}_1(c,s,a) = R_0$, where $R_0$ is the expected cumulative return calculated based on the only future context sequence (cf. the blue dashed lines in Figure \ref{fig:theorem2_illustration}a).
Notice that $\hat{Q}'(c,s,a)$ is calculated based on multiple context sequences but $\hat{Q}_1(c,s,a)$ is based on only one context sequence provided in the dataset.
Accordingly, $\hat{Q}_1(c,s,a)$ can suffer from a large variance.

An alternative method is to resort to the mapping between the high-dimensional context and the compact latent context.
When such a mapping is known, we can select the contexts that share a similar latent context to the interested context $c$ in the dataset (cf. the orange lines in Figure \ref{fig:theorem2_illustration}b).
Then, we can estimate the value function using $\hat{Q}_2(c,s,a) = \frac{1}{M} (R_0 + R_1 + R_2 + R_3 + R_4 + \cdots)$.
This estimate is the average of the returns collected based on multiple context sequences.
If these context sequences resemble the fictitious sequences starting from $(c,s,a)$ (cf. the gray dashed lines in Figure \ref{fig:theorem2_illustration}a),
we can consider $R_i\approx R'_i$ and therefore
$\hat{Q}_2(c,s,a)$ is an accurate estimate with reduced variance compared with $\hat{Q}_1$.

\begin{table*}[t]
    \centering
    \begin{tabular}{p{2cm}lrp{5cm}}
        \toprule
        Group & Variable & Trading Cost &  Explanation \\
        \midrule
        Revenue term &  \textbf{Cash inflow} & \textbf{-2.6455} & $r_1=n_t \bar{p}_t$ \\
        in reward & Price advantage & -2.5276 &  $n_t (\bar{p}_t - \bar{p}_\text{TWAP})$ \cite{fang2021universal} \\
        & Change in unrealized PnL & +2.6438 & See \cite{ning2018double} \\
        & Sparse reward & +2.5293 &  See \cite{lin2020end}\\
        \midrule
        Action space & Discrete volumes & +3.6614 & 21 discrete quoted volumes for MOs\textsuperscript{a} \\
        & Discrete prices & -0.1371 & 33 discrete quoted prices (relative)\textsuperscript{b} \\
        & \textbf{Discrete prices \& volumes} & \textbf{-3.5243} & 33$\times$4 discrete quoted price$\times$volume\textsuperscript{c} \\
        TWAP & $\beta=0.0$ & -0.0821 & The coefficient of the reg. term $r_2$ in reward\\
        regularization & $\mathbf{\beta=0.1}$ & \textbf{-1.7325} & \\
        & $\beta=1.0$ & +1.8146 & \\
        Model & Standard & +0.4229 & 6-layer MLP with 128 hidden neurons\\
        architecture & \textbf{Large} & \textbf{-0.4229} & 8-layer MPS with 256 hidden neurons\\
        Learning rate & $1\times 10^{-4}$ & +0.0347 & Learning rate used in training \\
        & $\mathbf{1\times 10^{-5}}$ & \textbf{-0.0347} & \\
        Market & LOB & +0.1526 & 5-level ask/bid volume/price\\
        variables& \textbf{LOB + factors} & \textbf{-0.1526} & + other LOB-based and technical factors\\
        \midrule
        Training mode & \textbf{Separate} & \textbf{-0.3154} & Train encoder, fix encoder, train policy network\\
        in Algo1 & Joint & +0.3154 & Train encoder and policy network simultaneously \\
        \bottomrule
    \end{tabular}

    \caption{
    Ablation study on different model designs for trade execution. 
    }
    
    \label{tab:ablation}
\end{table*}
\section{Simulator}
\label{app:execution_simulator}
\textbf{Dataset.}
Compared with the simulators based on the preset stochastic process \cite{daberius2019deep} or a collection of preset interactive agents \cite{byrd2019abides}, the simulator driven by real market data can capture the complex market more accurately \cite{vyetrenko2020get}.
Moreover, compared with the previous work where the simulator is based on bar-level data, our simulator relies on the LOB-level data of the market which records an LOB snapshot every 3 seconds.
With finer-grained data, we are able to learn more practical trading agents.
For example, we can evaluate how an agent that trades using only MOs suffers from a large trading cost, which is the scheme adopted in many existing papers \cite{fang2021universal}.
The interested time period of trade execution tasks in the industry is typically from 10 to 120 minutes, which is configurable in our simulator.
Our simulator is based on the dataset that records an LOB snapshot every 3 seconds from the real market.
The time period of trade execution tasks in our experiments is set to 30 minutes.
To avoid a long planning horizon, the agent interacts with the simulator at a lower frequency (i.e., one minute per step).
Nevertheless, simulation is carried out snapshot by snapshot for higher accuracy.

\textbf{Observations.}
The observation in our simulator consists of the private variable (i.e., the state) and the market variable (i.e., the context).
The private variable consists of the remaining time and executed quantity.
The market variable can be the stacked features (including order-book-related features, technical indicators, raw snapshots, etc.) over several past steps. 
Our simulator implements a wide range of features including the features that appear in the previous papers as far as we know.
For different designs on the observations space, the agent can choose from these features.
To eliminate the differences in the features on different stocks,
the simulator normalizes the features as follows:
The price (or the feature whose dimension is price) is normalized using z-score with the open price on that trading day as the mean and the volatility on the previous trading day as the standard deviation.
The volume (or the feature whose dimension is volume) is normalized by dividing by the total volume of the last trading day.
In specific algorithms, we may perform another normalization on these features to fit them into a proper value range.



\textbf{Actions.}
On each time step, our simulator receives a list of orders, each of which can be an MO or LO that specifies the direction, the quantity, and the price (only for LO).
On top of this, we provide a series of wrappers to fit different designs on the action space (e.g., discrete/continuous/combinatorial action spaces).
Our simulator will provide the best possible execution for each order.
For example, if the quoted price is lower than some bid price level, the simulator will automatically place an MO to fill the outstanding bid orders whose prices are higher than the quoted price, and an LO for the remaining quantity.
In our algorithms, the agent places an order on each step by choosing a quoted volume and a quoted price from discretized sets.
The quoted volume is selected from $\{\frac{1}{2}\text{TWAP}, \text{TWAP}, \frac{3}{2}\text{TWAP}, 2 \text{TWAP} \}$ where TWAP is the volume executed on each step by a TWAP strategy (i.e., selling an equal amount on each step).
The quoted price is specified by a price difference w.r.t. the best ask price. 
If the quoted price is lower than the best bid price, the agent actually places an MO; otherwise, it is an LO.
Outstanding orders at the end of each step will be withdrawn.

\textbf{Reward.}
The reward function may consist of a basic revenue term (e.g., negative trading cost or average execution price) and several regularization terms (e.g., approximating the permanent market impact or enforcing a TWAP-like strategy).
The revenue term reflects the overall objective of trade execution that minimizes the trading cost or the average execution price for a sell program.
Moreover, various of regularizers are adopted to model the permanent market impact or the prior knowledge of a good execution program (e.g., enforcing a TWAP-like program).
Our simulator provides various choices on the reward function design and can benchmark different designs with a uniform set of metrics (such as the trading cost or implementation shortfall \cite{perold1988implementation}).
In our algorithms, the reward function consists of a revenue term $r_1$ and a regularization term $r_2$, i.e., $r_t= r_1 + \beta r_2$ where $\beta$ is a coefficient.
The revenue term is $r_1 = n_t \bar{p}_t$ where
$n_t$ is the executed volume in the last step, and $\bar{p}_t$ is the corresponding average execution price.
The regularization term is $r_2 =  (v_t - v_{t,\text{TWAP}})^2$ where $v_t$ is the remaining inventory, and $v_{t,\text{TWAP}}$ is the remaining inventory if we follow the TWAP strategy.


\begin{figure*}[t]
     \centering
     \begin{subfigure}
         \centering
         \includegraphics[width=\textwidth]{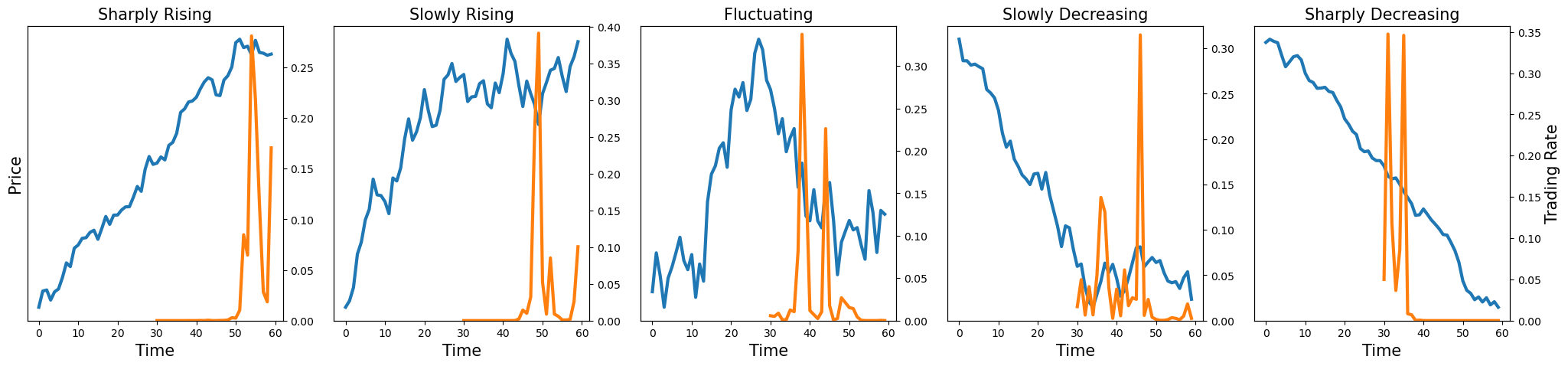}
     \end{subfigure}
     \hfill
     \begin{subfigure}
         \centering
         \includegraphics[width=\textwidth]{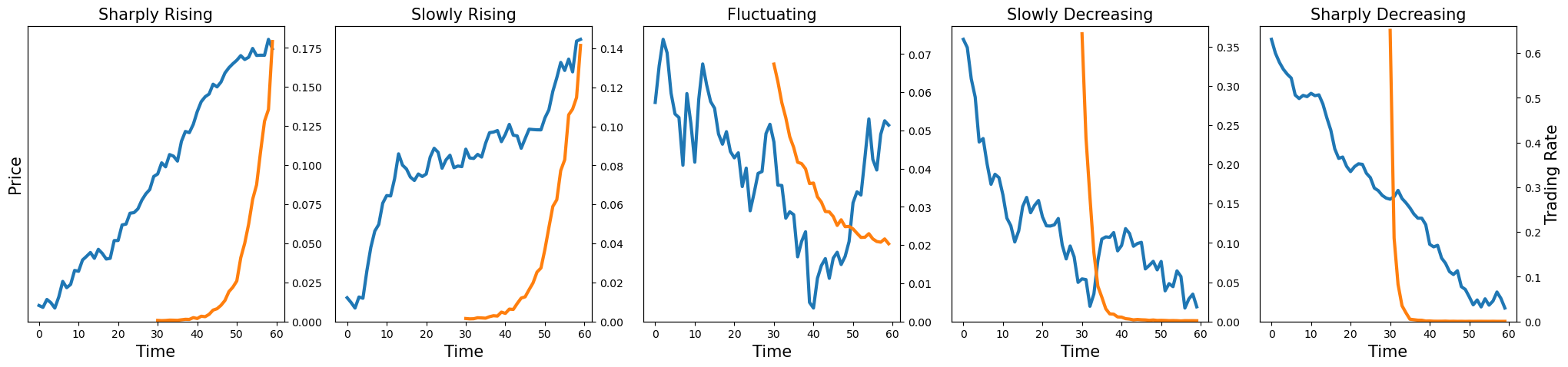}
     \end{subfigure}
     \caption{The static strategy learned by ``RL'' (top) and ``Fit parameters'' (bottom). The blue lines are the prices and the orange lines are the trading rates determined by the strategies. The price sequences in the five figures from left to right are generated using $\sigma=1.5$ and $\alpha=1.0, 0.5, 0.0, -0.5, -1.0$ respectively.}
     \label{fig:toy}
\end{figure*}

\textbf{Transition dynamics.}
Given a list of orders on the $t$-th time step, our simulator will determine the reward and the state on the next time step.
For MOs, we consider the temporary market impact and the time delay.
For example, when the decision of the agent is based on observation generated on time $\tau$, the execution of an MO is based on the snapshot on time $\tau+\Delta\tau$ where $\Delta\tau$ is a preset time delay.
Typically, we use $\Delta\tau=3s$.
For LOs, we determine whether the order can be executed snapshot by snapshot till the ($t+1$)-th time step.
If the highest market price (i.e., a transaction occurs on this price) in one snapshot exceeds the price quoted in the LO, we consider the order is fully executed.
If the highest market price exactly equals the quoted price, the order may be partially filled and the ratio is calculated by reconstructing the transactions between snapshots.
However, considering that 1) the quantity may be too large to be fully executed or 2) the LO may be at the end of the queue of the quoted price level, we impose additional trading limits on the above matching mechanism to encourage conservative simulation.

\textbf{Discussion.}
To improve fidelity, our simulator considers the temporary market impact of MOs, the time delay, and determines the execution of LOs based on reconstructing the transactions between snapshots.
However, there are still components that we do not consider.
First, the permanent market impact is the change of the equilibrium price during at least our planning horizon when we place an order.
Here, we assume the permanent market impact is linear w.r.t. the order quantity and therefore considering this factor does not change the optimal solution of our strategy \cite{almgren2001optimal}.
Then, MOs in our simulation not only change the current LOB but also the LOB of the next time step, possibly resulting in degenerated fidelity.
Therefore, we also rely on the assumption that the limit orders are resilient within a short period of time (which should be smaller or comparable to the time interval between two simulation steps).
Fortunately, this is verified by empirical studies such as \cite{degryse2005aggressive,cummings2010further,gomber2015liquidity}.


\begin{table*}[th] 
    \centering
    \begin{tabular}{lrr}
        \hline
        Algorithm & Testing & Gap \\
        \hline
        TWAP & 13.1217 (2.0858) & - \\
        \hline
        Tuned DQN & 6.7199 (2.7776) & 5.1121 \\
        \quad Nevmyvaka et al. & 9.7045 (1.2583) & 6.8319 \\
        \quad Ning et al. & 10.0825 (2.1145) & 2.8475 \\
        \quad Lin \& Beling & 12.4054 (1.3553) & 6.3749 \\
        \quad Tuned DQN + {CASH} & 3.9184 (1.2287) & \textbf{1.9833} \\
        \quad Tuned DQN + {CATE} & \textbf{0.0749} (1.5233) & 2.8275 \\
        \hline
        Tuned PPO & 3.2686 (0.5302) & 5.0026 \\
        \quad Dabérius et al. & 12.9987 (2.5579) & 6.2316\\
        \quad Lin \& Beling & 10.2599 (1.5293) & 4.4571 \\
        \quad Fang et al. & 11.5612 (5.3556) & 17.0537 \\
        \quad Tuned PPO + {CASH} & -3.8524 (0.2234) & \textbf{0.8558} \\
        \quad Tuned PPO + {CATE} & \textbf{-4.3000} (0.4813) & 0.8725 \\
        \hline
    \end{tabular}
    \caption{
    The trading cost (bp=$10^{-4}$) of different algorithms. 
    The numbers are the average mean (std.) trading cost in the last 100 evaluations of the total 1000 evaluations over five different random seeds.
    }
    \label{tab:other_periods}
\end{table*}

\begin{table*}[th]
    \centering
    \begin{tabular}{llllll}
    \hline
    Algorithm            & Training & Testing & Gap    \\ \hline
    Tuned DQN (3 months) & 2.0382 (1.7684)  & 5.9240 (3.2986)  & 3.8858 \\
    Tuned DQN (6 months) & 2.6027 (2.7152)  & 6.2365 (2.4883)  & 3.6338 \\
    CATE+DQN (3 months)  & -2.8774 (1.7019) & 0.0075 (1.6920)  & 2.8849 \\
    CATE+DQN (6 months)  & -3.1773 (1.6459) & -0.8044 (0.5289) & 2.3729 \\ \hline
    \end{tabular}
    \caption{Performance of the RL agents trained using more data.}
    \label{tab:more_data}
\end{table*}

\begin{table*}[th] 
    \centering
    \begin{tabular}{lrr}
    \hline
    Algorithm        & Testing & Gap \\ \hline
    Tuned DQN + PCA-CASH      & 5.3248 (2.9694)  & 1.8125  \\ 
    Tuned PPO + PCA-CASH      & 0.9382 (1.7338)  & 1.6602 \\ 
    \hline
    Tuned PPO + CASH (Ours)  &  -4.5760 (0.2062)  & \textbf{0.0744}
    \\
    Tuned PPO + CATE (Ours)  & \textbf{-4.9068 (0.3015)}  &  0.1296  \\ 
    \hline
    \end{tabular}
    \caption{The results of new baselines conducted under the same settings as Experiment 5.2.}
    \label{tab:other_baselines}
\end{table*}

\begin{table*}[th] 
    \centering
    \begin{tabular}{lcc}
    \hline
    Algorithm &  Daily return & Daily Sharpe        \\ 
    \hline
    Tuned DQN & -4.3380 (1.1141) & -- \\
    Tuned DQN + CASH & 1.1587 (0.9136) & 2.5794 (0.7398) \\
    Tuned DQN + CATE & \textbf{1.3619 (0.9721)} & \textbf{2.5849 (0.7593)} \\
    \hline
    \end{tabular}
    \label{tab:crypto}
    \caption{
    The performance of our algorithms in the trading task for quarterly BTC (cryptocurrency) features in OKEX (exchange).
    }
\end{table*}

\section{Further Experiment Details.}
\label{app:execution_experiments}

\textbf{Toy trade execution task.}
The input of the agent is a vector consists of past 30 steps prices changes. The output of the agent is a vector indicating a strategy that liquidates the inventory in the future 30 steps $a=(a_{31}, \cdots, a_{60})$ with $\sum_{t=31}^{60} a_t = 1$. 
The reward in the toy trade execution task is $r=(\sum_{t=31}^{60} \gamma^{t-31} a_t p_t) - (\frac{1}{30} \sum_{t=31}^{60} \gamma^{t-31} p_t)$, where $a_t$ is the action, i.e., selling a proportion of $a_t$ inventory on the $t$-th time step, $p_t$ is the price, and we set $\gamma=\exp(\log(\frac{1}{2}) / 30)= 0.9772$.
We present more results on the learned strategy by \emph{Base} and \emph{CASH} for the toy trade execution task to illustrate how an overfitted/generalizable policy should perform.
We present their performance under different market environments in Figure \ref{fig:toy}.
In general, when the price is rising (see the left two columns in the figure), a good strategy should sell more at the end of the horizon; otherwise (see the right two columns) a good strategy needs to liquidate as soon as possible. 
Moreover, \emph{Base} presents sharp peaks which indicates that it suffers from overfitting the training data while \emph{CASH} liquidates the inventory smoothly and resembles the analytical solution in previous papers \cite{almgren2001optimal}.

\textbf{Ablation study on the model designs for trade execution.}
We conduct experiments on the combination of the designs on the observation space, reward function, learning rate, model architecture, etc.
We perform a grid search over all the possible combinations and list the corresponding performance impacts in Table \ref{tab:ablation}.
We can observe that the most influential factor is the design in the reward function.
Tuned DQN/PPO and Algo1/Algo2 use the design that is the best in each group.

We conduct a grid search over all possible designs listed in the table and run each combination five times. 
We show the differences between the average trading cost when the model adopts the design and the average trading cost over the runs in the whole set.
For the first group, the whole set contains all runs for tuned PPO.
For the second to the fifth group, the whole set only contains runs for tuned PPO that use \emph{cash inflow} as the revenue term.
For the last group, the whole set contains all runs for Algo1 (PPO).
\begin{itemize}
\item Equally distributed quoted volumes ranges from 0 to 2TWAP as is in \cite{daberius2019deep}.
\item The quoted price is selected from a non-uniformly distributed set $P$=[[-50, -40, -30, -25, -20, -15], linspace(-10, 10, 21), [15, 20, 25, 30, 40, 50]] bp.
\item The quoted price and volume is selected from $[\frac{1}{2}\text{TWAP}, \text{TWAP}, \frac{3}{2}\text{TWAP}, 2 \text{TWAP}] \times P$.
\end{itemize}

\textbf{The experiment results at other time periods.}
To demonstrate the generalization of the model, we also conducted experiments at different time period. As shown in Table~\ref{tab:other_periods}, these results obtained on the same experiments settings as those in Section 5.2 but using data from another time period in 2021 (training: June 2021; validation: July 2021; testing: August 2021).

\textbf{Increasing the data volume.}
We argue that simply increasing the data volume does not address the overfitting problem since the RL agents are trained in the low-data region in practice.
Although there is a large amount of historical data available in finance, the high noise and high dimensionality of the context (i.e., the indicators) calls for far more samples than we have.
Moreover, in real trading, although using far data augments the dataset, this may also induce the distribution shift due to the fast changing market environment.
To validate our claim, we conduct new experiments using more data (6 months $\approx 3$M samples) and compare them with our main experiment. 
We present the result in Table \ref{tab:more_data}.
We observe that simply using more data does not alleviate overfitting and our algorithm achieves better performance when more data is consumed.

\textbf{Compare with existing context aggregation method.} 
To demonstrate the effectiveness of our algorithm in context aggregation, we designed a variant of \emph{CASH}, called \emph{PCA-CASH}, which utilizes principal component analysis (with \#components=8, same as \#statistics of \emph{CASH}) as the context aggregator. The experimental results are shown in Table \ref{tab:other_baselines} , we observed that our aggregation method demonstrates better performance.

\textbf{The experiment results in other markets.} 
We also conduct the experiment in the cryptocurrency market. The task is to decide to long or short for every 15 minutes. The reward is the profit gained during one time step. The evaluation metrics are average daily return (defined as the sum of profit in an average day assuming the initial asset is 1) and daily Sharpe (defined as the mean of daily returns divided by the standard deviation of daily returns). The transaction fee is set to be 0.03\%. In CASH, we hand-craft 7 statistics for the encoder to predict (on the future price trend, volatility, trading amount, etc.). In CATE, we train the encoder to generate 8-dimensional embeddings. The training data is from 2021-02-12 to 2021-08-13; the validation data is from 2021-08-14 to 2021-09-10; and the testing data is from 2021-09-10 to 2021-10-15. The numbers are the mean and standard deviation based on evaluations over 6 random seeds and the last 5 logged models during the training. We only conduct experiments on DQN due to the binary action space in this setting.

\textbf{Hardware.}
Our experiments are conducted on a server with the following configurations:
\begin{itemize}
    \item System: Ubuntu 18.04.5 LTS
    \item CPU: 24 $\times$ Intel(R) Xeon(R) CPU E5-2690 v4 @ 2.60GHz
    \item GPU: 4 $\times$ Tesla V100
    \item Memory: 441G
\end{itemize} 

\textbf{Limitation and social impact of our work.}
Generalization in trade execution or other tasks in quantitative investment is important. This paper only studies on how to improve generalization from the representation learning perspective. There are other aspects that are worth investing such as feature engineering, neural architecture design, and traditional deep learning techniques that can prevent overfitting.
Trade execution is an important application in the financial industry.
With a good trade execution strategy, the firms can allocate the resources more efficiently.



\section{Pseudo Code}
\label{app:algo}

In this section, we provide the pseudo code of CASH and CATE as supplementary to Figure~\ref{fig:algorithm}. We also provide the pseudo code for our simulator for better understanding the ORDC model illustrated in Figure~\ref{fig:ORDC}.

\begin{algorithm*}[th]
    \caption{CASH: Context Aggregation with Handcrafted Statistics}
    \label{algo:CASH_full}
    \begin{algorithmic}[1]
    \STATE Initialize context encoder $\phi_{\theta}(\cdot)$, PolicyNet/ValueNet $\text{RL}_{w}(\cdot,\cdot)$;
    \STATE Given hand-crafted statistics function $\phi_{\text{hand}}(\cdot)$, Simulated Environment $E$;
    
    \STATE \textit{\# Pre-train the context encoder}
    \FOR{$i = 1 $ \textbf{to} $ n$}
        \STATE Sample minibatch $\{(c_{j}, {c'}_{j})\}$ from simulator $E$ with arbitrary actions;
        \STATE Using gradient descent to minimize the loss $ \mathcal{L}_\text{stat}(\theta)= \sum_{j} (\phi_{\theta}(c_{j}) - \phi_{\text{hand}}({c'}_{j}))^{2}$;
    \ENDFOR
    \STATE \textit{\# Train the RL model}
    \STATE Fix the parameters $\theta$ in the context encoder $\phi_{\theta}(\cdot)$;
    \FOR{$i = 1 $ \textbf{to} $m$}
        \STATE Sample minibatch $\{(c_{j}, {s}_{j}, r_{j})\}$ from simulator $E$ by rolling the current policy;
        \STATE Train $\text{RL}_{w}(\phi_{\theta}(c_{j}), s_{j})$ by minimizing loss $\mathcal{L}(w)=\sum_{j} \mathcal{L}_\text{RL}(w)$
    \ENDFOR
    \end{algorithmic}
\end{algorithm*}

\begin{algorithm*}[th]
    \small
    \caption{CATE}
    \label{algo:CATE_full}
    \begin{algorithmic}[1]
    \STATE Initialize context encoder $\phi_{\theta}(\cdot)$, future encoder $\phi_{\vartheta}(\cdot)$, PolicyNet/ValueNet $\text{RL}_{w}(\cdot,\cdot)$;
    \STATE Given simulator $E$;
    
    \STATE \textit{\# Train model simultaneously in an end-to-end manner}
    \FOR{$i = 1 $ \textbf{to} $ m$}
        \STATE Sample minibatch $\{(c_{j}, {c}_{j}')\}$ from simulator $E$ by rolling the current policy;
        \STATE Using gradient descent to minimize the loss 
        \[
        \mathcal{L}(\theta, \vartheta, w):= \sum_{j} \mathcal{L}_{\text{Enc}}(\theta, \vartheta) + \mathcal{L}_{\text{RL1}}(\theta, w) + \mathcal{L}_{\text{RL2}}(\vartheta, w)
        \]
    \ENDFOR
    \end{algorithmic}
\end{algorithm*}

\begin{algorithm*}[th]
    \caption{Simulator}
    \label{algo:simulation}
    \begin{algorithmic}[1]
    \STATE \textbf{Input:} Policy $\pi$, 
    \STATE \textbf{Given:} Historical dataset $D$, transition dynamics $P(s'|c,s,a)$, and reward function $r(c,s,a)$, execution period $T$, and execution volume $V$;
    \STATE Pick a context trajectory $\{c_1, c_2, \cdots, c_T\}$ from $D$;
    \STATE Initialize $s_1$ and output $(c_1,s_1)$
    \FOR{$t = 1 $ \textbf{to} $ T$}
        \STATE Receive an action $a_t$ from $\pi$;
        \STATE Calculate $s_{t+1} \sim P(\cdot|c_t, s_t, a_t)$ and $r_{t} = r(c_t, s_t, a_t)$;
        \STATE Return $(r_t, c_{t+1}, s_{t+1})$;
    \ENDFOR
    \end{algorithmic}
\end{algorithm*}

\clearpage
\bibliographystyle{named}
\bibliography{reference}

\begin{thebibliography}{}

\bibitem[\protect\citeauthoryear{Agarwal \bgroup \em et al.\egroup
  }{2019}]{agarwal2019reinforcement}
Alekh Agarwal, Nan Jiang, Sham~M Kakade, and Wen Sun.
\newblock {\em {Reinforcement Learning: Theory and Algorithms}}.
\newblock CS Department, UW Seattle, Seattle, WA, USA, 2019.

\bibitem[\protect\citeauthoryear{Almgren and Chriss}{1999}]{almgren1999value}
Robert Almgren and Neil Chriss.
\newblock Value under liquidation.
\newblock {\em Risk}, 12(12):61--63, 1999.

\bibitem[\protect\citeauthoryear{Almgren and Chriss}{2001}]{almgren2001optimal}
Robert Almgren and Neil Chriss.
\newblock Optimal execution of portfolio transactions.
\newblock {\em Journal of Risk}, 3:5--40, 2001.

\bibitem[\protect\citeauthoryear{Arjovsky \bgroup \em et al.\egroup
  }{2019}]{arjovsky2019invariant}
Martin Arjovsky, L{\'e}on Bottou, Ishaan Gulrajani, and David Lopez-Paz.
\newblock Invariant risk minimization.
\newblock {\em arXiv preprint arXiv:1907.02893}, 2019.

\bibitem[\protect\citeauthoryear{Azar \bgroup \em et al.\egroup
  }{2013}]{azar2013minimax}
Mohammad~Gheshlaghi Azar, R{\'e}mi Munos, and Hilbert~J Kappen.
\newblock {Minimax PAC bounds on the sample complexity of reinforcement
  learning with a generative model}.
\newblock {\em Machine Learning}, 91(3):325--349, 2013.

\bibitem[\protect\citeauthoryear{Bulthuis \bgroup \em et al.\egroup
  }{2017}]{bulthuis2017optimal}
Brian Bulthuis, Julio Concha, Tim Leung, and Brian Ward.
\newblock Optimal execution of limit and market orders with trade director,
  speed limiter, and fill uncertainty.
\newblock {\em International Journal of Financial Engineering}, 4(1):175--200,
  2017.

\bibitem[\protect\citeauthoryear{Byrd \bgroup \em et al.\egroup
  }{2019}]{byrd2019abides}
David Byrd, Maria Hybinette, and Tucker~Hybinette Balch.
\newblock {ABIDES: Towards high-fidelity market simulation for AI research}.
\newblock {\em arXiv preprint arXiv:1904.12066}, 2019.

\bibitem[\protect\citeauthoryear{Cobbe \bgroup \em et al.\egroup
  }{2019}]{cobbe2019quantifying}
Karl Cobbe, Oleg Klimov, Chris Hesse, Taehoon Kim, and John Schulman.
\newblock Quantifying generalization in reinforcement learning.
\newblock In {\em International Conference on Machine Learning}, pages
  1282--1289. PMLR, 2019.

\bibitem[\protect\citeauthoryear{Cobbe \bgroup \em et al.\egroup
  }{2020}]{cobbe2020leveraging}
Karl Cobbe, Chris Hesse, Jacob Hilton, and John Schulman.
\newblock Leveraging procedural generation to benchmark reinforcement learning.
\newblock In {\em International conference on machine learning}, pages
  2048--2056. PMLR, 2020.

\bibitem[\protect\citeauthoryear{Cummings and
  Frino}{2010}]{cummings2010further}
James~Richard Cummings and Alex Frino.
\newblock Further analysis of the speed of response to large trades in interest
  rate futures.
\newblock {\em Journal of Futures Markets: Futures, Options, and Other
  Derivative Products}, 30(8):705--724, 2010.

\bibitem[\protect\citeauthoryear{Dab{\'e}rius \bgroup \em et al.\egroup
  }{2019}]{daberius2019deep}
Kevin Dab{\'e}rius, Elvin Granat, and Patrik Karlsson.
\newblock Deep execution-value and policy based reinforcement learning for
  trading and beating market benchmarks.
\newblock {\em Available at SSRN 3374766}, 2019.

\bibitem[\protect\citeauthoryear{Degryse \bgroup \em et al.\egroup
  }{2005}]{degryse2005aggressive}
Hans Degryse, Frank~De Jong, Maarten~Van Ravenswaaij, and Gunther Wuyts.
\newblock Aggressive orders and the resiliency of a limit order market.
\newblock {\em Review of Finance}, 9(2):201--242, 2005.

\bibitem[\protect\citeauthoryear{Dietterich \bgroup \em et al.\egroup
  }{2018}]{dietterich2018discovering}
Thomas Dietterich, George Trimponias, and Zhitang Chen.
\newblock Discovering and removing exogenous state variables and rewards for
  reinforcement learning.
\newblock In {\em International Conference on Machine Learning}, pages
  1262--1270. PMLR, 2018.

\bibitem[\protect\citeauthoryear{Du \bgroup \em et al.\egroup
  }{2019}]{du2019provably}
Simon Du, Akshay Krishnamurthy, Nan Jiang, Alekh Agarwal, Miroslav Dudik, and
  John Langford.
\newblock {Provably efficient RL with rich observations via latent state
  decoding}.
\newblock In {\em Proceedings of the 36th International Conference on Machine
  Learning}, pages 1665--1674. PMLR, 2019.

\bibitem[\protect\citeauthoryear{Fang \bgroup \em et al.\egroup
  }{2021}]{fang2021universal}
Yuchen Fang, Kan Ren, Weiqing Liu, Dong Zhou, Weinan Zhang, Jiang Bian, Yong
  Yu, and Tie-Yan Liu.
\newblock Universal trading for order execution with oracle policy
  distillation.
\newblock In {\em Proceedings of the AAAI Conference on Artificial
  Intelligence}, volume~35, pages 107--115, 2021.

\bibitem[\protect\citeauthoryear{Gomber \bgroup \em et al.\egroup
  }{2015}]{gomber2015liquidity}
Peter Gomber, Uwe Schweickert, and Erik Theissen.
\newblock {Liquidity dynamics in an electronic open limit order book: An event
  study approach}.
\newblock {\em European Financial Management}, 21(1):52--78, 2015.

\bibitem[\protect\citeauthoryear{Gu{\'e}ant \bgroup \em et al.\egroup
  }{2012}]{gueant2012optimal}
Olivier Gu{\'e}ant, Charles-Albert Lehalle, and Joaquin Fernandez-Tapia.
\newblock Optimal portfolio liquidation with limit orders.
\newblock {\em SIAM Journal on Financial Mathematics}, 3(1):740--764, 2012.

\bibitem[\protect\citeauthoryear{Kirk \bgroup \em et al.\egroup
  }{2021}]{kirk2021survey}
Robert Kirk, Amy Zhang, Edward Grefenstette, and Tim Rockt{\"a}schel.
\newblock A survey of generalisation in deep reinforcement learning.
\newblock {\em arXiv preprint arXiv:2111.09794}, 2021.

\bibitem[\protect\citeauthoryear{Lin and Beling}{2020}]{lin2020deep}
Siyu Lin and Peter~A Beling.
\newblock A deep reinforcement learning framework for optimal trade execution.
\newblock In {\em Joint European Conference on Machine Learning and Knowledge
  Discovery in Databases}, pages 223--240. Springer, 2020.

\bibitem[\protect\citeauthoryear{Lin and Beling}{2021}]{lin2020end}
Siyu Lin and Peter~A Beling.
\newblock An end-to-end optimal trade execution framework based on proximal
  policy optimization.
\newblock In {\em Proceedings of the Twenty-Ninth International Conference on
  International Joint Conferences on Artificial Intelligence}, pages
  4548--4554, 2021.

\bibitem[\protect\citeauthoryear{Mao \bgroup \em et al.\egroup
  }{2017}]{mao2017neural}
Hongzi Mao, Ravi Netravali, and Mohammad Alizadeh.
\newblock Neural adaptive video streaming with pensieve.
\newblock In {\em Proceedings of the Conference of the ACM Special Interest
  Group on Data Communication}, pages 197--210, 2017.

\bibitem[\protect\citeauthoryear{Mao \bgroup \em et al.\egroup
  }{2018}]{mao2018variance}
Hongzi Mao, Shaileshh~Bojja Venkatakrishnan, Malte Schwarzkopf, and Mohammad
  Alizadeh.
\newblock Variance reduction for reinforcement learning in input-driven
  environments.
\newblock In {\em International Conference on Learning Representations}, 2018.

\bibitem[\protect\citeauthoryear{Nevmyvaka \bgroup \em et al.\egroup
  }{2006}]{nevmyvaka2006reinforcement}
Yuriy Nevmyvaka, Yi~Feng, and Michael Kearns.
\newblock Reinforcement learning for optimized trade execution.
\newblock In {\em Proceedings of the 23rd international conference on Machine
  Learning}, pages 673--680. PMLR, 2006.

\bibitem[\protect\citeauthoryear{Ning \bgroup \em et al.\egroup
  }{2018}]{ning2018double}
Brian Ning, Franco Ho~Ting Lin, and Sebastian Jaimungal.
\newblock {Double deep Q-learning for optimal execution}.
\newblock {\em arXiv preprint arXiv:1812.06600}, 2018.

\bibitem[\protect\citeauthoryear{Oroojlooyjadid \bgroup \em et al.\egroup
  }{2022}]{oroojlooyjadid2022deep}
Afshin Oroojlooyjadid, MohammadReza Nazari, Lawrence~V Snyder, and Martin
  Tak{\'a}{\v{c}}.
\newblock {A deep Q-network for the beer game: Deep reinforcement learning for
  inventory optimization}.
\newblock {\em Manufacturing \& Service Operations Management}, 24(1):285--304,
  2022.

\bibitem[\protect\citeauthoryear{Packer \bgroup \em et al.\egroup
  }{2018}]{packer2018assessing}
Charles Packer, Katelyn Gao, Jernej Kos, Philipp Kr{\"a}henb{\"u}hl, Vladlen
  Koltun, and Dawn Song.
\newblock Assessing generalization in deep reinforcement learning.
\newblock {\em arXiv preprint arXiv:1810.12282}, 2018.

\bibitem[\protect\citeauthoryear{Perold}{1988}]{perold1988implementation}
A.~F. Perold.
\newblock {The implementation shortfall: Paper vs. reality}.
\newblock {\em Journal of Portfolio Management}, 14(3):4--9, 1988.

\bibitem[\protect\citeauthoryear{Shahamiri}{2008}]{shahamiriname2008reinforcement}
Masoud Shahamiri.
\newblock Reinforcement learning in environments with independent delayed-sense
  dynamics.
\newblock {\em Master Thesis, University of Alberta}, 2008.

\bibitem[\protect\citeauthoryear{Shen \bgroup \em et al.\egroup
  }{2020}]{shen2020auxiliary}
Wei Shen, Xiaonan He, Chuheng Zhang, Qiang Ni, Wanchun Dou, and Yan Wang.
\newblock Auxiliary-task based deep reinforcement learning for participant
  selection problem in mobile crowdsourcing.
\newblock In {\em Proceedings of the 29th ACM International Conference on
  Information \& Knowledge Management}, pages 1355--1364, 2020.

\bibitem[\protect\citeauthoryear{Silver \bgroup \em et al.\egroup
  }{2014}]{silver2014deterministic}
David Silver, Guy Lever, Nicolas Heess, Thomas Degris, Daan Wierstra, and
  Martin Riedmiller.
\newblock Deterministic policy gradient algorithms.
\newblock In {\em Proceedings of the 31st International Conference on Machine
  Learning}, pages 387--395. PMLR, 2014.

\bibitem[\protect\citeauthoryear{Song \bgroup \em et al.\egroup
  }{2019}]{song2019observational}
Xingyou Song, Yiding Jiang, Stephen Tu, Yilun Du, and Behnam Neyshabur.
\newblock Observational overfitting in reinforcement learning.
\newblock {\em arXiv preprint arXiv:1912.02975}, 2019.

\bibitem[\protect\citeauthoryear{Vyetrenko \bgroup \em et al.\egroup
  }{2020}]{vyetrenko2020get}
Svitlana Vyetrenko, David Byrd, Nick Petosa, Mahmoud Mahfouz, Danial Dervovic,
  Manuela Veloso, and Tucker Balch.
\newblock {Get real: Realism metrics for robust limit order book market
  simulations}.
\newblock In {\em Proceedings of the First ACM International Conference on AI
  in Finance}, pages 1--8, 2020.

\bibitem[\protect\citeauthoryear{Wang \bgroup \em et al.\egroup
  }{2020}]{wang2020improving}
Kaixin Wang, Bingyi Kang, Jie Shao, and Jiashi Feng.
\newblock Improving generalization in reinforcement learning with mixture
  regularization.
\newblock {\em Advances in Neural Information Processing Systems},
  33:7968--7978, 2020.

\bibitem[\protect\citeauthoryear{Zhang \bgroup \em et al.\egroup
  }{2018a}]{zhang2018dissection}
Amy Zhang, Nicolas Ballas, and Joelle Pineau.
\newblock A dissection of overfitting and generalization in continuous
  reinforcement learning.
\newblock {\em arXiv preprint arXiv:1806.07937}, 2018.

\bibitem[\protect\citeauthoryear{Zhang \bgroup \em et al.\egroup
  }{2018b}]{zhang2018study}
Chiyuan Zhang, Oriol Vinyals, Remi Munos, and Samy Bengio.
\newblock A study on overfitting in deep reinforcement learning.
\newblock {\em arXiv preprint arXiv:1804.06893}, 2018.

\end{thebibliography}

\end{document}